\documentclass[journal]{IEEEtran}
\IEEEoverridecommandlockouts
\usepackage[utf8]{inputenc}
\usepackage{colortbl} 
\usepackage[T1]{fontenc}
\usepackage{textcomp}
\usepackage{fixltx2e}
\usepackage{microtype}
\usepackage{calc}
\usepackage[normalem]{ulem}
\usepackage{makecell}
\usepackage{diagbox}
\usepackage{pifont}
\usepackage{gensymb}
\usepackage{comment}
\usepackage{bbm}
\usepackage{amsthm}
\newtheorem{proposition}{Proposition}

\usepackage{amsmath,amssymb,amsfonts}
\usepackage{algorithmic}
\usepackage{multirow}
\usepackage{textcomp}
\usepackage{xcolor}\usepackage[dvipsnames]{xcolor}
\usepackage{url}
\usepackage{ tipa }
\usepackage[range-phrase=--,per-mode=symbol-or-fraction,binary-units=true,range-units=single,list-units=single,detect-all]{siunitx}
\usepackage[capitalize,noabbrev]{cleveref}
\crefformat{footnote}{#2\footnotemark[#1]#3}

\usepackage{booktabs}
\usepackage{url}
\RequirePackage{xstring}
\RequirePackage{xparse}
\RequirePackage[]{acro}
\NewDocumentCommand\acrodef{mO{#1}mG{}}{\DeclareAcronym{#1}{short={#2}, long={#3}, #4}}
\usepackage[ruled,linesnumbered]{algorithm2e}
\usepackage{siunitx}

\DeclareUnicodeCharacter{0301}{\'{e}}
\acrodef{mmWave}{millimeter-wave}
\acrodef{ISAC}{integrated sensing and communications}
\acrodef{PMCW}{phase-modulated continuous waveform}
\acrodef{CO-PMCW}{code-orthogonal PMCW}
\acrodef{CSO-PMCW}{code-space-orthogonal PMCW}
\acrodef{FMCW}{frequency-modulated continuous waveform}
\acrodef{OFDM}{orthogonal frequency-division multiplexing}
\acrodef{MIMO}{multiple-input multiple-output}
\acrodef{V2X}{vehicle-to-everything}
\acrodef{PAPR}{peak-to-average power ratio}
\acrodef{SINR}{signal-to-interference-plus-noise ratio}
\acrodef{RDM}{range-Doppler map}
\acrodef{SISO}{single-input single-output}
\acrodef{CRLB}{Cramer-Rao lower bound}
\acrodef{ULA}{uniform linear array}
\acrodef{UE}{user equipment}
\acrodef{PRNS}{pseudo-random noise sequence}
\acrodef{AoD}{angle of departure}
\acrodef{AoA}{angle of arrival}
\acrodef{AWGN}{additive white Gaussian noise}
\acrodef{RCS}{radar cross section}
\acrodef{RSU}{roadside unit}
\acrodef{LoS}{line-of-sight}
\acrodef{NLoS}{non-line-of-sight}
\acrodef{FFT}{fast Fourier transform}
\acrodef{IFFT}{inverse fast Fourier transform}
\acrodef{PRI}{pulse repetition interval}
\acrodef{i.i.d.}{independent and identically distributed}
\acrodef{MU-MIMO}{multi-user MIMO}
\acrodef{FIM}{Fisher information matrix}
\acrodef{IIoT}{intelligent internet of things}
\acrodef{LTE}{long-term evolution}
\acrodef{NR}{new radio}
\acrodef{C-RAN}{cloud radio access network}
\acrodef{RRA}{radio resource allocation}
\acrodef{MSE}{mean squared error}
\acrodef{VA}{virtual array}
\acrodef{EFIM}{equivalent Fisher information matrix}
\acrodef{TDM}{time-division multiplexing}
\acrodef{RMS}{root mean square}
\acrodef{FDM}{frequency-division multiplexing}
\acrodef{CDM}{code-division multiplexing}
\acrodef{OTFS}{orthogonal time frequency space}
\acrodef{PEB}{position error bound}
\acrodef{VEB}{velocity error bound}
\acrodef{MMS}{multi-monostatic sensing}
\acrodef{MBS}{multi-bistatic sensing}
\acrodef{MXS}{multi-X-static hybrid sensing}
\acrodef{UAV}{unmanned aerial vehicle}
\acrodef{BS}{base station}
\acrodef{SNR}{signal-to-noise ratio}
\acrodef{eMBB}{enhanced mobile broadband}
\acrodef{MUSIC}{multiple signal classification}
\acrodef{CS}{compressive sensing}
\acrodef{FE}{fine estimation}
\acrodef{Tx}{transmitter}
\acrodef{Rx}{receiver}
\acrodef{TRx}{transceiver}
\acrodef{DL}{downlink}
\acrodef{UP}{uplink}
\acrodef{DoF}{degree of freedom}
\acrodef{ESNR}{energy SNR}
\acrodef{BPSK}{binary phase shift keying}
\acrodef{RC}{raised cosine}
\acrodef{ICI}{inter-carrier interference}
\acrodef{RRC}{root-raised cosine}
\acrodef{CP}{cyclic prefix}
\acrodef{ISFFT}{inverse symplectic finite Fourier transform}
\acrodef{SFFT}{symplectic finite Fourier transform}
\acrodef{DD}{delay-Doppler}
\acrodef{BTDM}{block TDM}
\acrodef{ITDM}{interleaved TDM}
\acrodef{SIMO}{single-input multiple-output}
\acrodef{BFDM}{block FDM}
\acrodef{CFDM}{comb FDM}
\acrodef{MRT}{maximum ratio transmission}
\acrodef{ISI}{inter-symbol interference}
\acrodef{RF}{radio frequency}

\usepackage[caption=false,font=footnotesize]{subfig}
\usepackage[pdftex]{graphicx}
\DeclareGraphicsExtensions{.pdf,.png,.jpg,.tikz}
\usepackage{csquotes}
\usepackage[backend=biber,style=ieee,doi=false,isbn=false,sorting=none,sortcites=true,mincitenames=1,maxcitenames=2]{biblatex}
\DeclareFieldFormat{sentencecase}{#1} % never apply sentence casing, even if bibtex field is unprotected
\DeclareFieldFormat{titlecase}{#1} % never apply title casing, even if bibtex field is unprotected
\usepackage{xpatch}
\xpatchbibmacro{textcite}{\addspace}{\addnbspace}{}{}
\xpatchbibmacro{Textcite}{\addspace}{\addnbspace}{}{}
\DefineBibliographyStrings{english}{
	andothers = et~al\adddot\addspace
}

\graphicspath{{./Bilder/}}    
\begin{document}

\title{On Unified CRLB Framework from Generic Signals to ISAC Waveforms with Virtual Array Sensing}

\author{%
Yanpeng Su,~\IEEEmembership{Graduate Student Member,~IEEE}, Norman Franchi,~\IEEEmembership{Member,~IEEE},\\and Maximilian Lübke,~\IEEEmembership{Member,~IEEE}
\thanks{The work contributes to the research within the 6G-Valley innovation cluster. \textit{(Corresponding author: Yanpeng Su.)}}
\thanks{Yanpeng Su, Norman Franchi, and Maximilian Lübke are with the Institute for Smart Electronics and Systems, Friedrich-Alexander-Universität Erlangen-Nürnberg, 91058 Erlangen, Germany (email: yanpeng.su@fau.de; norman.franchi@fau.de; maximilian.luebke@fau.de).}
\thanks{Color versions of one or more of the figures in this article are available online at http://ieeexplore.ieee.org.}
}
\markboth{ }
{Su \textit{et al.}: On Unified CRLB Framework from Generic Signals to ISAC Waveforms with Virtual Array Sensing}
%{Su \textit{et al.}: Cooperative ISAC Networks: When Is the Tractable CRLB Accurate for Joint Position and Velocity Estimation?}

\maketitle

\begin{abstract}
This paper presents a unified Cramér–Rao lower bound (CRLB) framework for signal-level parameters in integrated sensing and communications (ISAC)-enabled radar systems.
Starting from the generic signal model, we analyze the coupling between delay and Doppler in the Fisher information matrix (FIM), which is unsolved and often overlooked in relevant studies.
Addressing this issue, we derive the conditions under which the coupling terms can be eliminated and demonstrate that these conditions are typically satisfied for ISAC-enabled waveforms. 
Afterward, the CRLBs of representative ISAC waveforms are derived within the unified framework, enabling consistent and comparable analysis across the waveforms and avoiding model-dependent discrepancies. %Moreover, the delay-Doppler decoupling provides a tractable CRLB as a metric for system optimization.
Further, the framework is extended to virtual array (VA) sensing systems, where the impact of different multiplexing schemes is analyzed. 
Simulation results demonstrate the consistency between the CRLBs derived from the proposed framework and those obtained from waveform-specific analyses. The proposed framework shows strong generality, waveform-compatibility, and flexibility, offering a versatile tool for the CRLB analysis of various waveforms, including those lacking existing analytical results.
\end{abstract}


\section{Introduction}
Radar sensing is envisioned to be integrated into the upcoming 6G mobile networks since it allows the networks to perceive the environments, accelerating a wide range of applications, including intelligent transportation systems, smart home and cities, and \ac{IIoT} \cite{10536135,persp}. 
%Meanwhile, sidelink-based sensing is also expected to play a crucial role in 6G networks, since they enable the vehicles to be aware of their surroundings and thereby accelerate the development of autonomous driving and efficient traffic management \cite{9127852,10772915}. 
Since the deployment of the sensing function potentially leads to increased cost and resource scarcity, the coexistence and cooperation between sensing and communication systems becomes a particularly important issue. In this context, \ac{ISAC} has been widely discussed in recent years and is expected to be a key technology of 6G \cite{10781422,10561589}. By allowing the communication and sensing functions to share the spectral and hardware, ISAC improves resource efficiency and reduces hardware and energy costs \cite{9924202,9737357}. 

%Although the 4G \ac{LTE} and 5G \ac{NR} networks already provide localization services for active devices like smartphones \cite{9354629}, the above-mentioned applications require future 6G wireless networks to be able to sense both active and passive objects \cite{9591277}. 

While the performance bounds for communication systems are well established from 2G to 5G, the fundamental limitations of radar sensing have been widely discussed in recent years. 
Performance bound analysis is the prerequisite for integrating sensing into mobile networks.
On the one hand, it offers guidance for networks in the geometrical deployment and parameter configuration. On the other hand, it serves as an important metric in \ac{RRA} and beamforming and clustering design. 
In this context, \ac{CRLB} provides a fundamental limitation on sensing accuracy in terms of radio resource and \ac{SNR}. It calculates the theoretical lower bound on the \ac{MSE} of any unbiased estimator \cite{1420803}. 

In radar systems, CRLB is generally used to evaluate the achievable estimation accuracy of the signal-level parameters, including amplitude, phase, delay (range), Doppler (radial velocity), and \ac{AoA}, and the state parameters, including position and velocity. Accurate CRLB estimation relies on the joint analysis of all unknown parameters in the \ac{FIM}, since ignoring coupling between parameters can lead to inaccurate analysis. The CRLB of the parameters of interest can be calculated by the \ac{EFIM}, which is obtained as the Schur complement with respect to the nuisance parameters in the FIM, thereby capturing their effect on the parameters of interest. This work addresses gaps in signal-level CRLB and the impact of \ac{VA} and provides a unified framework for CRLB analysis. From the perspective of position and velocity estimation, range, velocity, and AoA are considered as the parameters of interest, while the amplitude and phase are nuisance parameters. 

\subsection{Related Works}
Related studies of signal-level CRLB cover the analysis of the generic signal model and specific waveforms, as well as the extension to the \ac{VA}-enabled systems. 
The study in \cite{923295} provides a comprehensive CRLB estimation for the signal-level parameters, where the coupling between delay and Doppler in the FIM is described in a non-closed form, which is demonstrated to significantly increase the difficulty of calculating the CRLB. 
On the contrary, the AoA is decoupled from other parameters.
Studies in \cite{7418183,7347470} analyze the CRLB using the Taylor series and demonstrate that the delay and Doppler CRLBs are respectively proportional to the square of bandwidth and frame length, while the closed-form expression is missing. In \cite{5494404,6034675}, the authors designed a bistatic radar channel selection algorithm based on the range and radial velocity CRLBs, where the influence of nuisance parameters is not considered. Power allocation and beamforming designs employing CRLB as the sensing metric can be found in \cite{9652071,10978368,10097000}, where the narrowband signal model is usually considered, abandoning range information and the estimation of radial velocity. In summary, an accurate, closed-form, and tractable CRLB of the generic signal models is missing. 

The coupling between delay and Doppler in generic signal models is still not solved. Recent studies concentrate more on the analysis of specific waveforms. The CRLB of \ac{FMCW} is provided in \cite{7472072,9489355}. 
Studies in \cite{10414289,9695370} discuss the CRLB of \ac{PMCW} signals, while a closed-form and tractable expression for a wideband PMCW signal is missing.
Besides, it is worth noting that the shaped pulses can influence the shapes of the chips and the spectrum of PMCW, potentially impacting the CRLB, which is not addressed yet. 
The closed-form CRLB of \ac{OFDM} radar can be found in \cite{braun2014ofdm,11231051,10706865}. 
Further, \cite{10373881} extends the application scenario of OFDM radar to sidelink position estimation. 
The CRLB of \ac{OTFS} is analyzed in \cite{8757044}, where the provided formula of CRLB shows a high complexity, while the simulation result demonstrates that the CRLBs of OTFS and OFDM radars are almost the same.
The impact of VA technology on the CRLB of AoA is discussed in \cite{li2021thinned,7300572,6638411}. However, since VA generally relies on signal multiplexing schemes such as \ac{TDM}, \ac{FDM}, and \ac{CDM}, the reduced resource and non-coherent transmission scheme may impact the CRLBs, whose effect has not been investigated.

\subsection{Contributions}
Compared to waveform-specific analyses, a generic CRLB provides a unified framework that encompasses different waveform-specific cases as special instances, enabling systematic and fair performance comparison across signal designs and providing an efficient tool for CRLB analysis of waveforms whose results are lacking. 
To this end, it is shown to be beneficial to solve the coupling issues between delay and Doppler and provide a tractable and unified approach for CRLB analysis as a common benchmark for various waveforms. Besides, the resulting tractable formula provides a useful metric for optimization problems.
Further, VA is a favorable technology in collocated MIMO radar since it improves spatial resolution and target separability, while its impact on the CRLB has not been thoroughly investigated.

This paper addresses these gaps by providing a unified CRLB framework for signal-level parameters.
%Starting from the FIM of the generic signal model, we analyze the conditions under which the delay-Doppler coupling can be eliminated or approximately neglected. Particularly, we investigate the coupling issue under different waveforms and derive their CRLBs under the unified framework. Afterward, we extend the framework to systems with VA and various multiplexing schemes, and the corresponding changes in FIM and CRLB are analyzed. 
%Simulation focuses on comprehensive performance evaluation and demonstration of the consistency between the CRLBs derived from the proposed framework and those obtained in waveform-specific analyses.
The main contributions of this paper are summarized as follows:
\begin{itemize}
    \item \textbf{Generic signal-level analysis and conditions for decoupling:} Starting from the generic signal model, we first review the derivation of signal-level FIM with all parameters. %including delay, Doppler, AoA, amplitude, and phase. %The EFIM of the parameters of interest is then obtained by the Schur complement with respect to the nuisance parameters. 
    Addressing the coupling term between delay and Doppler, we present the conditions under which they can be eliminated or approximately neglected.
    The generic signal-based analysis and proposed decoupling conditions are compatible with a wide range of waveforms, requiring only the substitution of their specific structures.
    This leads to a unified CRLB framework that provides a general representation capable of capturing diverse waveform characteristics.
    \item \textbf{Waveform analysis:} We address four representative ISAC waveforms in this work: FMCW, PMCW with various shaped pulses, OFDM, and OTFS. The conditions for decoupling under these waveforms are analyzed, and their CRLBs are calculated based on the generic signal model with the integration of their characteristics. For OFDM and OTFS, we show that 
    the CRLB derived from the proposed framework has consistent expression with those obtained from waveform-specific analyses in related works, demonstrating its correctness. For PMCW, we fill the gap in its closed-form CRLBs and analyze the behavior of different shaped pulses.
    \item \textbf{Impact of VA and multiplexing:} We analyze the changes in CRLB produced by VA and multiplexing schemes. First, we reconstruct a collocated MIMO radar transmission model and demonstrate that the overall FIM is the sum of the individual FIMs corresponding to each multiplexed transmit signal. Afterward, the CRLBs with different multiplexing schemes are derived. The result indicates that the non-coherent transmission and multiplexing are the dominant factors resulting in differences from phased array radars. 
    \item \textbf{Performance evaluation and correctness verification:} The simulation results demonstrate the consistency between the proposed CRLB calculation approach and that obtained by waveform-specific analyses. Based on the results, the generality, flexibility, and waveform-compatibility of the presented framework are explained. %The proposed unified framework provides an efficient tool for the CRLB analysis of various waveforms.
\end{itemize}

The rest of this paper is organized as follows: Section~\ref{sec.model} introduces the system model. The FIM and CRLB are analyzed in Section~\ref{sec.crlb}, followed by the analysis of decoupling conditions and the investigation of different waveforms. Section~\ref{sec.VA} extends the framework to VA sensing with different multiplexing schemes. The numerical results are given in Section~\ref{sec.simulation}. Section~\ref{sec.conclusion} summarizes this work.

\textit{Notations:}  $(\cdot)^{*}$, $(\cdot)^{T}$, $(\cdot)^{H}$, and $(\cdot)^{-1}$ denote the conjugate, transpose, conjugate transpose, and inverse, respectively. $a$ and $A$ represent scalars, $\mathbf{a}$ and $\mathbf{A}$ represent vectors and matrices, respectively. 
The diagonalization calculation of vector $\mathbf{a}$ is written by $\mathrm{diag}$($\mathbf{a}$). $\odot$ and $\otimes$ denote the Hadamard product and Kronecker product.
%$\mathbf{A}\odot \mathbf{B}$ and $\mathbf{A}\otimes \mathbf{B}$ imply Hadamard product and outer product of $\mathbf{A}$ and $\mathbf{B}$. 
%$\mathbf{A}[n,m]$ implies the component at the $n$-th row and $m$-th column of matrix $\mathbf{A}$. 
%$\mathbf{A}(m\downarrow)$ and $\mathbf{A}(n\rightarrow)$ express the $m$-th column and the $n$-th row of matrix $\mathbf{A}$, respectively. 
$\mathbb{Z},\ \mathbb{N},\ \mathbb{R},$ and $\mathbb{C}$ express the sets of integers, natural numbers, real numbers, and complex numbers. $\mathcal{R}\{\cdot\}$ and $\mathcal{I}\{\cdot\}$ extract the real and imaginary parts.
$\mathcal{CN}(\mu,\sigma^2)$ denotes the complex Gaussian distribution with mean and variance of $\mu$ and $\sigma^2$. $s(t)\leftrightarrow S(f)$ implies Fourier transform. $\dot{s}(t)=ds(t)/dt$ denotes the derivative of $s(t)$. For simplification, $\int_{-\infty}^\infty(\cdot)dt$ is abbreviated as $\int(\cdot)dt$.

\section{System Model}\label{sec.model}
This work exemplarily showcases a traffic scenario, where the planar position and velocity are involved. 
Hence, we consider a multiple-antenna radar system where the \ac{Tx} and \ac{Rx} antennas are \acp{ULA} with sizes of $N_\text{T}$ and $N_\text{R}$, respectively, and their inter-element spacing is denoted by $d_\text{T}$ and $d_\text{R}$.
%The CRLB of radar sensing with phased array antennas is analyzed in this section and Section~\ref{sec.crlb}, while the results of \ac{VA} sensing are discussed in Section~\ref{sec.VA}. For phased-array antennas, 
A common configuration is $d_\text{T}=d_\text{R}=\lambda/2$.
Beginning with the generic signal model $s(t)$, the received signal at the Rx is given by
\begin{align}
    \mathbf{r}(t)=\underbrace{Ae^{j\phi}\mathbf{a}_{\text{R}}^*(\theta_{\text{R}}){s}(t-\tau)e^{j2\pi f_\text{D}t}}_{\pmb{\mu}(t)}+\mathbf{n}(t),\label{eq.rec}
\end{align}
where $A$, $\phi$, $\tau$, and $f_\text{D}$ denote the amplitude, phase, delay, and Doppler shift. $A=a\mathbf{a}_{\text{T}}^H(\theta_{\text{T}})\mathbf{w}=N_\text{T}a$ combines the path loss and beamforming gain, where the \ac{MRT} with $\mathbf{w}=\mathbf{a}_{\text{T}}(\theta_{\text{T}})$ is considered in this work. \Ac{AoD} $\theta_{\text{T}}$ is not included in FIM analysis since it does not introduce new \ac{DoF} of independent observation as $\mathbf{a}_{\text{T}}^H(\theta_{\text{T}})\mathbf{w}$ is scalar and cannot lead to additional observation dimensions. $\mathbf{a}_{\text{R}}(\theta_{\text{R}})$ denotes the Rx steering vector. $\mathbf{n}(t)\sim\mathcal{CN}(0,\sigma^2)^{N_\text{R}\times1}$ represents \ac{AWGN} whose power is assumed to be the same for all Rx antenna elements.

The likelihood and log-likelihood functions of the signal-level parameters $\pmb{\Theta}=[A,\phi,\tau,f,\theta_\text{R}]^T$ are given by 
\begin{align}
    &f(\mathbf{r}|\pmb{\Theta})=\frac{1}{(\pi\sigma^2)^{N_\text{R}}}\exp\Big(-\frac{1}{\sigma^2}\int||\mathbf{r}(t)-\pmb{\mu}(t)||^2dt\Big)\\
    &\log f(\mathbf{r}|\pmb{\Theta})=-\frac{1}{\sigma^2}\int||\mathbf{r}(t)-\pmb{\mu}(t)||^2dt,
\end{align}
where the constant term $-N_\text{R}\log(\pi\sigma^2)$ in log-likelihood function is omitted since its derivative to $\pmb{\Theta}$ is 0.

\section{Derivation of CRLB}\label{sec.crlb}
This section addresses the signal-level CRLB in radar systems.
The FIM of the parameters $\pmb{\Theta}$ based on the generic signal model was addressed in previous works like \cite{923295}, whereas the coupling between range and radial velocity is not solved. In this section, we first review this issue and then analyze the conditions under which the coupling disappears. Afterward, the CRLBs of different waveforms are derived based on the obtained generic model.

The FIM of $\pmb{\Theta}$ is calculated by
\begin{align}
    \mathbf{F}_{i,j}=-\mathbb{E}\Big\{\frac{\partial^2\log f(\mathbf{r}|\pmb{\Theta})}{\partial{\Theta}_i\partial{\Theta}_j}\Big\}.
\end{align}

The derivation of FIM $\mathbf{F}$ and EFIM $\mathbf{E}$ is provided in Appendix~\ref{sec.SLF}. The resulting EFIM is given in (\ref{eq.efim}),
\begin{figure*}\begin{align}
    \mathbf{E}=\begin{bmatrix}
        2N_\text{R}\gamma\Big(4\pi^2B_\text{rms}^2+\frac{C_0^2}{E_\text{s}^2}\Big) & \frac{4\pi A^2N_\text{R}\mathcal{I}\{C_1\}}{\sigma^2} & 0\\
        \frac{4\pi A^2N_\text{R}\mathcal{I}\{C_1\}}{\sigma^2} & 8\pi^2N_\text{R}\gamma T_\text{rms}^2 &0\\
        0 & 0 & \frac{\pi^2\cos^2(\theta_\text{R})N_\text{R}(N_\text{R}^2-1)\gamma}{6}
    \end{bmatrix}.\label{eq.efim}
\end{align}\end{figure*}
where \ac{ESNR} $\gamma=\frac{A^2P_\text{s}T_\text{F}}{\sigma^2}=\frac{A^2E_\text{s}}{\sigma^2}$ denotes the overall energy ratio between signal and AWGN at each Rx antenna element, and $P_\text{s}$ is the average signal power. $B_\text{rms}$ and $T_\text{rms}$ are the \ac{RMS} bandwidth and time:
\begin{align}
    B_\text{rms}^2=\frac{\int f^2|S(f)|^2df}{\int |S(f)|^2df},\quad T_\text{rms}^2=\frac{\int t^2|s(t)|^2dt}{\int |s(t)|^2dt},\label{eq.rms}
\end{align}
where $S(f)$ is the spectrum of $s(t)$. For signals with uniformly distributed power over bandwidth $B$ and frame $T_\text{F}$, $B_\text{rms}^2={B^2}/{12}$, $T_\text{rms}^2={T_\text{F}^2}/{12}$, where $B$ and $T_\text{F}$ represent the signal bandwidth and frame length.

The undesired terms $C_0$ and $C_1$ are given by
\begin{align}
    &C_0=\int s^*(t-\tau)\dot{s}(t-\tau)dt\label{eq.c0},\\
    &C_1=\int t\dot{s}^*(t-\tau)s(t-\tau)dt\label{eq.c1}.
\end{align}
%Note $\dot{s}(t-\tau)$ denotes the derivative of $s(t-\tau)$ to $\tau$ instead of $t$
The CRLB is calculated by extracting the diagonal elements of $\mathbf{E}^{-1}$, i.e.,
\begin{align}
    \mathbf{C}=\mathbf{E}^{-1},\ C_\tau=\mathbf{C}_{1,1},\ C_{f_\text{D}}=\mathbf{C}_{2,2},\ C_{\theta_\text{R}}=\mathbf{C}_{3,3},
\end{align}
where ${\theta_\text{R}}$ is decoupled from other parameters and independent of waveforms:
\begin{align}
    C_{\theta_\text{R}}=\frac{6}{\pi^2\cos^2(\theta_\text{R})N_\text{R}(N_\text{R}^2-1)\gamma}.
\end{align}
Previous works \cite{923295,7418183,7347470,rs12182913} have provided the EFIM in (\ref{eq.efim}) and the expression of $C_0$ and $\mathcal{I}\{C_1\}$ in (\ref{eq.c0}) and (\ref{eq.c1}). However, a closed-form solution for $C_0$ and $\mathcal{I}\{C_1\}$ is still not provided, which is a prerequisite for deriving a clear and tractable expression of $C_\tau$ and $C_{f_\text{D}}$. A tractable expression is desired since it can offer a clear form of the transformation to state parameters (position and velocity) and can be employed as the target function in optimization problems, such as beamforming and \ac{RRA}.
The conditions under which $C_0$ and $\mathcal{I}\{C_1\}$ can be neglected are analyzed in this work:

\begin{proposition}\label{p.0}
    The undesired term $C_0$ vanishes when the signal power spectrum is an even function, i.e., symmetric about zero frequency, $|S(f)|^2=|S(-f)|^2$.
\end{proposition}
\begin{proof}
    Since $\dot{s}(t)\leftrightarrow j2\pi f S(f)$, according to the Parseval's theorem, we have
    \begin{align}
        &C_0=\int s^*(t-\tau)\dot{s}(t-\tau)dt=\int s^*(u)\dot{s}(u)du\nonumber\\
        &=\int j2\pi fS^*(f)S(f)df=j2\pi \int f|S(f)|^2df.
    \end{align}
    If the power spectrum $|S(f)|^2$ is symmetric about $f=0$, $C_0=0$.
\end{proof}

\begin{proposition}\label{p.1}
    For signals $s(t)$ with power spectrum symmetric about zero frequency, the coupling term $\mathcal{I}\{C_1\}$ vanishes if the derivative of the phase of $S(f)$ is an even function of $f$, or equivalently, the derivative of the phase of $s(t)$ is symmetric to the energy centroid in the time domain. 
\end{proposition}
\begin{proof}
    For any signal $s(t)$, it can be expressed by $s(t)=p(t)e^{j\psi(t)}$, where $p(t)$ and $\psi(t)$ respectively represent the amplitude and phase of $s(t)$. $C_1$ can be reformulated as
    \begin{align}
        &C_1=\int t\dot{s}^*(t-\tau)s(t-\tau)dt=\int (u+\tau)\dot{s}^*(u)s(u)du\nonumber\\
        &=\negthinspace\int\negthinspace u(\dot{p}(u)\negthinspace-\negthinspace jp(u)\dot{\psi}(u))e^{-j\psi(u)}p(u)e^{j\psi(u)}du\negthinspace-\negthinspace\int\negthinspace \tau C_0\nonumber\\
        &\overset{\text{Symmetric }|S(f)|^2}{=}\int (u\dot{p}(u)p(u)-ju{p}^2(u)\dot{\psi}(u))du,
    \end{align}
    \begin{align}
        \mathcal{I}\{C_1\}=-\int u{p}^2(u)\dot{\psi}(u)du=-\int u|s(u)|^2\dot{\psi}(u)du.
    \end{align}
    As mentioned in Appendix~\ref{sec.SLF}, the power centroid is shifted to the origin. If $\dot{\psi}(u)$ is an even function, we have $\mathcal{I}\{C_1\}=0$.

    Similarly, in the frequency domain, since $ts(t)\leftrightarrow\frac{j}{2\pi}\dot{S}(f)$, according to the Parseval's theorem, $C_1$ can be expressed as
    \begin{align}
        &C_1=\int u\dot{s}^*(u)s(u)du=\int fS^*(f)\dot{S}(f)df,\\
    &\mathcal{I}\{C_1\}=\int f|S(f)|^2\dot{\Psi}(f)df.
    \end{align}
    If the derivative of phase in spectrum $\dot{\Psi}(f)$ is an even function, we have $\mathcal{I}\{C_1\}=0$.
\end{proof}

Generally, the condition in \textbf{Proposition}~\ref{p.0} holds since baseband signals usually have (approximately) symmetric power spectrum, while the condition in \textbf{Proposition}~\ref{p.1} is more complicated. A special case is for real signals, since $\psi(t)\in\{0,\pi\}$, $\dot{\psi}(t)=0$ always holds except for zero crossing points, where the contribution is 0 due to the zero amplitude. In the following, we analyze the values of $C_0$ and $\mathcal{I}\{C_1\}$ with representative ISAC waveforms, derive their CRLBs based on the generic signal model, and verify the correctness in the comparison with waveform-specific results.
%Different from related waveform-specific studies, where the CRLB is derived from specific signal structure, we directly include the waveform characteristics into the generic signal model, benefiting from a simpler and directly comparable analytical framework across different waveforms. 
%Since the signal power of the discussed waveforms is uniformly distributed, $T_0=T_\text{F}/2$ is considered in the following.

\subsection{FMCW}
The derivation of FMCW CRLB from the generic signal model has been provided in \cite{923295,9489355} and is briefly reviewed here. The transmit signal is given by
\begin{align}
    s(t)=\sum_{k=0}^{K-1}x_ke^{j\pi\mu(t-kT+T_0)^2}\mathrm{rect}\Big(\frac{t-kT+T_0}{T}\Big),
\end{align}
where $\mu=B/T$ denotes the ratio between bandwidth and \ac{PRI}. $x_k$ is the carried data.
$K=T_\text{F}/T$ represents the number of PRIs in a radar frame. $T_0=T_\text{F}/2$ is defined in Appendix~\ref{sec.SLF}. $C_0$ and $\mathcal{I}\{C_1\}$ are calculated by
\begin{align}
    &\negthinspace C_0\negthinspace\approx\negthinspace\int \negthinspace\sum_{k=0}^{K-1}\negthinspace 2j\pi\mu |x_k|^2(t\negthinspace-\negthinspace kT\negthinspace+\negthinspace{T_0})\mathrm{rect}\Big(\frac{t\negthinspace-\negthinspace kT+T_0}{T}\Big)dt\nonumber\\
    &\overset{u=t-kT+T_0}{=}2j\pi\mu\sum_{k=0}^{K-1}|x_k|^2\int_{-\frac{T}{2}}^{\frac{T}{2}}  udu=0,
\end{align}
\begin{align}
    &\mathcal{I}\{C_1\}=-\int t|s(t)|^2\dot{\psi}(t)dt\nonumber\\
    &=-2\pi\mu\sum_{k=0}^{K-1}|x_k|^2\int t\,\mathrm{rect}\Big(\frac{t-kT+T_0}{T}\Big)\Big(t-kT+T_0\Big)dt\nonumber\\
    &=-2\pi\mu\sum_{k=0}^{K-1}|x_k|^2\int_{-\frac{T}{2}}^{\frac{T}{2}} u(u+kT-T_0)dt\nonumber\\
    &=-\frac{\pi\mu T^3}{6}\sum_{k=0}^{K-1} |x_k|^2=-\frac{\pi\mu T^2E_\text{s}}{6},
\end{align}
where the derivative of the rectangular function at the edges is neglected for analytical tractability.

Substitute $C_0$ and $\mathcal{I}\{C_1\}$ into (\ref{eq.efim}), the CRLBs of delay and Doppler are calculated by
\begin{align}\label{eq.crlbfmcw}
    &C_\tau=\frac{3}{2\pi^2N_\text{R}\gamma B^2\Big(1-\frac{1}{K^2}\Big)}\overset{K\gg1}{\approx}\frac{3}{2\pi^2N_\text{R}\gamma B^2},\nonumber\\
    &C_{f_\text{D}}=\frac{3}{2\pi^2N_\text{R}\gamma T_\text{F}^2\Big(1-\frac{1}{K^2}\Big)}\overset{K\gg1}{\approx}\frac{3}{2\pi^2N_\text{R}\gamma T_\text{F}^2},
\end{align}
where the condition $K\gg1$ generally holds for the functionality of radial velocity estimation.

\subsection{PMCW}
To our best knowledge, a closed-form expression for PMCW CRLB is still lacking, while the impact of shaped pulses on PMCW has not been discussed.
The PMCW signal with \ac{PRNS} length of $L$ is written as
\begin{align}
    &s(t)=\sum_{k=0}^{K-1}\sum_{l=0}^{L-1}x_kb_lg(t-lT_\text{c}-kT+T_0), \\
    &S(f)=G(f)\sum_{k=0}^{K-1}\sum_{l=0}^{L-1}x_kb_le^{-j2\pi f(lT_\text{c}+kT-T_0)}\\
    &\negthinspace|\negthinspace S\negthinspace(f)\negthinspace|\negthinspace^2\negthinspace=\negthinspace|\negthinspace G\negthinspace(f)\negthinspace|^2\Big|\negthinspace\sum_{k=0}^{K-1}\negthinspace\sum_{l=0}^{L-1}\negthinspace x_kb_le^{\negthinspace-j2\pi f(lT_\text{c}+kT)}\negthinspace\Big|^2\negthinspace=\negthinspace|G(f)|^2|B(f)|^2\negthinspace,\negthinspace
\end{align}
where $g(t)$ denotes the shaped pulse. 
$x_k$ and $b_l$ represent the embedded communication data and the chip of PRNS. $T_\text{c}$ is the chip period.
The power spectrum of common pulses, such as rectangular, sinc, \ac{RRC}, and \ac{RC} impulses, is even. And $|B(f)|^2$ is also even since
\begin{align}
   &|B(f)|^2=\sum_{k,k'=0}^{K-1}\sum_{l,l'=0}^{L-1}x_kx_{k'}^*b_lb_{l'}^*e^{-j2\pi f((l-l')T_\text{c}+(k-k')T)}\nonumber\\
    &%=\negthinspace\sum_{k,k'=0}^{K-1}\sum_{l,l'=0}^{L-1}b_kb_{k'}^*b_lb_{l'}^*e^{-j2\pi f((l'-l)T_\text{c}+(k'-k)T)}
    =|B(-f)|^2\negthinspace,
\end{align}
where for each pair of unequal $k=a,k'=b$, there always exists another pair of $k=b,k'=a$, the same for $l$ and $l'$. %Hence $|B(f)|^2=|B(-f)|^2$ holds.
According to \textbf{Proposition}~\ref{p.0}, $C_0=0$ for PMCW signals. 

The \acp{PRNS} such as m-sequences, Gold codes, Kasami codes, etc., are binary, and the communication data symbols in PMCW ISAC systems are generally \ac{BPSK}-modulated to maintain the characteristics of binary radar systems \cite{9627227}, thus $x_k,b_l=\pm1,\,\forall k,l$, the condition $\psi(t)=\dot{\psi}(t)=0$ always holds for common shaped pulses. According to \textbf{Proposition}~\ref{p.1}, $\mathcal{I}\{C_1\}=0$.

However, while $T_\text{rms}^2=T_\text{F}^2/12$ holds, $B_\text{rms}^2$ changes with various shaped pulses. For randomly generated $x_k$ and $b_l$, $\mathbb{E}\{x_kb_lx_{k'}^*b_{l'}^*\}=\delta[l-l']\delta[k-k']$, where $\delta[\cdot]$ denotes the Kronecker delta function.
Consequently, $|B(f)|^2$ approximates flat when $K,L\gg1$ due to the law of large numbers:
\begin{align}
    &|B(f)|^2\overset{K,L\gg1}{\approx}\mathbb{E}\{|B(f)|^2\}\nonumber\\
    &\negthinspace=\negthinspace\sum_{k,k'\negthinspace=0}^{K-1}\sum_{l,l'\negthinspace=0}^{L-1}\negthinspace\mathbb{E}\negthinspace\{x_kx_{k'}^*\}\mathbb{E}\negthinspace\{b_lb_{l'}^*\}e\negthinspace^{-j2\pi f((l-l')T_\text{c}+(k-k')T)}\negthinspace=\negthinspace KL.\negthinspace
\end{align}
Therefore, $B_\text{rms}^2$ is only influenced by $|G(f)|^2$, yielding $B_\text{rms}^2=B_{\text{rms},G(f)}^2$. The RMS bandwidths of common shaped pulses are derived in Appendix~\ref{sec.brms}. The delay and Doppler CRLBs of PMCW are then calculated by:
\begin{align}
    &C_\tau=\begin{cases}
        \frac{4.448}{\pi^2N_\text{R}\gamma B^2}=\frac{0.4507}{N_\text{R}\gamma B^2}, & \text{Rect.}\\
        \frac{3}{2\pi^2N_\text{R}\gamma B^2}, & \text{Sinc}\\
        \frac{3}{2\pi^2N_\text{R}\gamma B^2}\negthinspace\cdot\negthinspace\frac{\pi^2(1+\alpha)^2}{(3\pi^2-24)\alpha^2+\pi^2}, & \text{RRC}\\
        \frac{3}{2\pi^2N_\text{R}\gamma B^2}\negthinspace\cdot\negthinspace\frac{\pi^2(4-\alpha)(1+\alpha)^2}{(6-\pi^2)\alpha^3+(12\pi^2-96)\alpha^2-3\pi^2\alpha+4\pi^2}, & \text{RC}
    \end{cases}\nonumber\\
    &C_{f_\text{D}}=\frac{3}{2\pi^2N_\text{R}\gamma T_F^2},\label{eq.pmcwf}
\end{align}
where $\alpha$ denotes the roll-off factor.

\subsection{OFDM}
The OFDM spectrum and the signal with \ac{CP} are given by
\begin{align}
    &\negthinspace S(f)\negthinspace=\negthinspace T\negthinspace\sum_{k=0}^{K-1}\negthinspace\sum_{l=0}^{L-1}\negthinspace X_{l,k}\mathrm{sinc}\Big(\frac{f+f_0}{\Delta f}-l\Big)e^{-j2\pi f(kT_\text{s}-T_0)}\negthinspace.\negthinspace\label{eq.ofdmsf}\\
    &\negthinspace s(t)\negthinspace=\negthinspace\sum_{k=0}^{K-1}\negthinspace\sum_{l=0}^{L-1}\negthinspace X\negthinspace_{l,k}e\negthinspace^{j2\pi{(l\Delta\negthinspace f\negthinspace-\negthinspace f_0)\negthinspace(t\negthinspace-\negthinspace kT\negthinspace_\text{s}\negthinspace+\negthinspace T\negthinspace_0\negthinspace-\negthinspace T\negthinspace_\text{cp}\negthinspace+\negthinspace T)}}\negthinspace\mathrm{rect}\negthinspace \Big(\negthinspace \frac{t\negthinspace-\negthinspace kT\negthinspace_\text{s}\negthinspace+\negthinspace T_0\negthinspace}{T_\text{cp}}\negthinspace \Big)\negthinspace   \nonumber\\
    &\negthinspace+\negthinspace\sum_{k=0}^{K-1}\negthinspace\sum_{l=0}^{L-1}\negthinspace X\negthinspace_{l,k}e\negthinspace^{j2\pi{(l\Delta\negthinspace f\negthinspace-\negthinspace f_0)\negthinspace(t\negthinspace-\negthinspace kT\negthinspace_\text{s}\negthinspace+\negthinspace T\negthinspace_0\negthinspace-\negthinspace T\negthinspace_\text{cp})}}\negthinspace\mathrm{rect}\negthinspace \Big(\negthinspace \frac{t\negthinspace-\negthinspace kT\negthinspace_\text{s}\negthinspace+\negthinspace T\negthinspace_0\negthinspace-\negthinspace T\negthinspace_\text{cp}\negthinspace}{T}\negthinspace \Big)\negthinspace.\negthinspace\label{eq.ofdmst}
\end{align}
where $\Delta f=1/T=B/L$ denotes the subcarrier distance, $f_0=\Delta f(L-1)/2$ guarantees that the frequency center is 0. $L$ in OFDM implies the number of subcarriers. $T_\text{s}=T+T_\text{cp}$ is the OFDM symbol length, $T_\text{cp}=TL_\text{cp}/L$ denotes the \ac{CP} length. 
%The \ac{CP} is omitted since it is removed before Rx radar signal estimation. 
The data symbols $X_{l,k}\in\mathbf{X}$ are assumed to be \ac{i.i.d.} %(or at least second-order uncorrelated with zero-mean and identical variance) 
with zero mean and variance of $P_X$, yielding $\mathbb{E}\{X_{k,l}X_{k',l'}^*\}=P_X\delta[k-k']\delta[l-l']$, which is a standard assumption in communication-enabled signals. 

The mathematical expectation of the power spectrum is symmetric due to the symmetric subcarriers
\begin{align}
    %&|S(f)|^2=T^2\sum_{k,k'=0}^{K-1}\sum_{l,l'=0}^{L-1}X_{l,k}X_{l',k'}^*\mathrm{sinc}\Big(\frac{f}{\Delta f}-l+\frac{L-1}{2}\Big)\nonumber\\
    %&\qquad\qquad\cdot\mathrm{sinc}\Big(\frac{f}{\Delta f}-l'+\frac{L-1}{2}\Big)e^{-j2\pi (k-k')T}\\
    &\negthinspace\mathbb{E}\{\negthinspace|S(f)|^2\negthinspace\}\negthinspace=\negthinspace T^2\negthinspace\sum_{k,k'=0}^{K-1}\sum_{l,l'=0}^{L-1}\negthinspace\mathbb{E}\{X_{l,k}\negthinspace X_{\negthinspace l',k'\negthinspace}^*\}\mathrm{sinc}\Big(\frac{f+f_0}{\Delta f}\negthinspace-\negthinspace l\Big)\nonumber\\
    &\qquad\qquad\qquad\cdot\mathrm{sinc}\Big(\frac{f+f_0}{\Delta f}-l'\Big)e^{-j2\pi f(k-k')T_\text{s}}\nonumber\\
    %&={T^2}\sum_{k=0}^{K-1}\sum_{l=0}^{L-1}P_X\mathrm{sinc}^2\Big(\frac{f+f_0}{\Delta f}-l\Big)\nonumber\\
    &={T^2}KP_X\sum_{l=0}^{L-1}\mathrm{sinc}^2\Big(\frac{f+f_0}{\Delta f}-l\Big),\\
    &\negthinspace\mathbb{E}\{|S(-f)|^2\}\negthinspace=\negthinspace {T^2}\negthinspace K\negthinspace P_X\negthinspace\sum_{l=0}^{L-1}\negthinspace\mathrm{sinc}\negthinspace^2\negthinspace\Big(\negthinspace\frac{f-f_0}{\Delta f}\negthinspace+\negthinspace l\Big)\negthinspace=\negthinspace\mathbb{E}\{\negthinspace|S(f)|^2\negthinspace\}.\negthinspace
\end{align}
In the general case, such as 5G-based mobile networks, $K,L\gg1$. Following the law of large numbers, $|S(f)|^2\approx\mathbb{E}\{|S(f)|^2\}$, i.e., the power spectrum is approximately symmetric. According to \textbf{Proposition}~\ref{p.0}, $C_0\approx0$. 

Since the OFDM signal is complex, a closed-form of $\dot{\psi}(t)$ in \textbf{Proposition}~\ref{p.1} cannot be obtained. Hence, $\mathcal{I}\{C_1\}$ is derived in a closed-form here.
However $\mathcal{I}\{C_1\}$ can be solved with zero mean and i.i.d. $X_{l,k}$ and large $K,L$:
\begin{align}
    %&\dot{s}^*(t)=-\frac{1}{L}\sum_{k=0}^{K-1}\sum_{l=0}^{L-1}X_{l,k}^*j2\pi (l\Delta f-f_0)e^{-j2\pi{(l\Delta f-f_0)(t-kT)}}\mathrm{rect}\Big(\frac{t-kT}{T}\Big)\\
    &\mathbb{E}\{\mathcal{I}\{C_1\}\negthinspace \}\negthinspace =\negthinspace -{j2\pi}\negthinspace \int \negthinspace t\negthinspace \sum_{k,k'=0}^{K-1}\sum_{l,l'=0}^{L-1}\negthinspace \mathbb{E}\{X_{l,k}^*X_{l',k'}\} (l\Delta f\negthinspace -\negthinspace f_0)\nonumber\\
    &\cdot\negthinspace\Big(\negthinspace e\negthinspace^{ - j2\pi{(l\Delta\negthinspace f - f_0)\negthinspace (t\negthinspace - kT_\text{s}+T_0-T_\text{cp}+T)}} \negthinspace e\negthinspace^{j2\pi{(l'\negthinspace \Delta\negthinspace f - f_0)\negthinspace(t - k'T_\text{s}+T_0-T_\text{cp}+T)}}\negthinspace\nonumber\\
    &\cdot \text{rect}\big( \frac{t - kT_\text{s}+ T_0}{T_\text{cp}} \big) \text{rect}\big( \frac{t -  k'T_\text{s}+ T_0}{T_\text{cp}} \big)\nonumber\\
    &+e^{ - j2\pi{(l\Delta f - f_0) (t\negthinspace - kT_\text{s}+T_0-T_\text{cp})}}  e^{j2\pi{(l' \Delta f - f_0)(t - k'T_\text{s}+T_0-T_\text{cp})}}\nonumber\\
    &\cdot \text{rect}\big( \frac{t - kT_\text{s}+ T_0-T_\text{cp}}{T} \big) \text{rect}\big( \frac{t -  k'T_\text{s}+ T_0-T_\text{cp}}{T} \big)\Big)
    dt\nonumber
    \\
    &=\negthinspace-{j2\pi P_X}\negthinspace\sum_{k=0}^{K-1}\negthinspace\sum_{l=0}^{L-1}\negthinspace (l\Delta f-f_0)\negthinspace\int\negthinspace t\,\text{rect}^2\Big(\frac{t-kT_\text{s}+T_0}{T_\text{s}}\Big)=0,
\end{align}
where the cross terms between the useful signal and CP are canceled by time isolation.
$\sum_{l=0}^{L-1} (l\Delta f-f_0)=0$ and time centralization guarantee the final result is 0. 
In summary, when $K,L\gg1$ and the elements in $\mathbf{X}$ are zero-mean i.i.d. or more generally, uncorrelated with zero-mean and identical variance, $C_0\approx0$ and $\mathcal{I}\{C_1\}\approx\mathbb{E}\{\mathcal{I}\{C_1\}\}=0$. This property is inherent in modern mobile networks. 

The power is uniformly distributed in the time and frequency domains, hence $B_\text{rms}^2=B^2/12$, $T_\text{rms}^2=T_\text{F}^2/12$, where $T_\text{F}=KT_\text{s}$ instead of $KT$ due to the inserted CP. The CRLB of OFDM can be calculated by
\begin{align}
    &C_\tau\approx\frac{3}{2\pi^2N_\text{R}\gamma B^2},\ C_{f_\text{D}}\approx\frac{3}{2\pi^2N_\text{R}\gamma T_F^2}.\label{eq.ofdmcrlb}
\end{align}

Considering the multi-carrier characteristic, another perspective on calculating the RMS bandwidth and time is by leveraging the structure in the discrete domain, i.e.,
\begin{align}
    &\int f^2|S(f)|^2df= \sum_{k=0}^{K-1}\sum_{l=0}^{L-1}|X_{k,l}|^2\Big(\big(l-\frac{L-1}{2}\big)\Delta f\Big)^2\nonumber\\
    &\approx \frac{KL(L^2-1)\Delta f^2P_X}{12},\nonumber\\
    &\int t^2|s(t)|^2dt\approx \sum_{k=0}^{K-1}\sum_{l=0}^{L-1}|X_{k,l}|^2\Big(\big(k-\frac{K-1}{2}\big)T_\text{s}\Big)^2\nonumber\\
    &\approx \frac{KL(K^2-1)T_\text{s}^2P_X}{12},
\end{align}
where the approximation neglects the cross terms between $X_{l,k}$ and $X_{l',k'}$. This results in the same expression as that derived in waveform-specific studies \cite{11231051,10706865}:
\begin{align}\label{eq.crlbofdmnew}
    &C_\tau'\approx\frac{3}{2\pi^2N_\text{R}\gamma_XKL(L^2-1)\Delta f^2}=\frac{L^2}{L^2-1} C_\tau,\nonumber\\ &C_{f_\text{D}}'\approx\frac{3}{2\pi^2N_\text{R}\gamma_XKL(K^2-1)T_\text{s}^2}=\frac{K^2}{K^2-1} C_{f_\text{D}},
\end{align}
where $\gamma_X=P_X/\sigma^2=\frac{\gamma}{KL}$ denotes the SNR per symbol. The analyses in the continuous and discrete domains show a high consistency. The derivation based on the generic model avoids the high complexity of the waveform-specific approach, which requires computing and summing each $X_{l,k}$ and eliminating coupling effects.

\subsection{OTFS}\label{sec.otfs}
The OTFS signal has the same expression as OFDM in (\ref{eq.ofdmsf})-(\ref{eq.ofdmst}), whereas $\mathbf{X}$ is the \ac{ISFFT} of the \ac{DD} domain symbols $\mathbf{x}$, i.e.,
\begin{align}
    &X_{l,k}=\frac{1}{\sqrt{KL}}\sum_{\nu=0}^{K-1}\sum_{\mu=0}^{L-1}x_{\mu,\nu}e^{j2\pi\big(\frac{k\nu}{K}-\frac{l\mu}{L}\big)},% \mathbf{X}=\mathbf{F}_L\mathbf{x}\mathbf{F}_K^H,
\end{align}
%where $\mathbf{F}_L=\frac{1}{\sqrt{L}}\Big[e^{-j2\pi\frac{\mu\nu}{L}}\Big]_{0\le\mu\le L-1\atop0\le\nu\le L-1}$ denotes the \ac{FFT} matrix, and $\mathbf{F}_K^H$ denotes the \ac{IFFT} matrix. 
where the ISFFT maps the symbols $\mathbf{x}$ from the DD domain to the time-frequency domain.
Note that the communication symbols are carried by $\mathbf{x}$, and the elements in $\mathbf{X}$ are hence not independent of each other.
However, the DD domain symbols $\mathbf{x}$ are i.i.d. with zero mean, hence 
\begin{align}
    &\mathbb{E}\{X_{k,l}X_{k',l'}^*\}\nonumber\\
    &\negthinspace=\negthinspace\frac{1}{\negthinspace K L\negthinspace}\negthinspace\sum_{\nu,\nu'=0}^{K-1}\negthinspace\sum_{\mu,\mu'=0}^{L-1}\negthinspace\mathbb{E}\{\negthinspace x_{\mu,\nu}x_{\mu',\nu'}^*\negthinspace\}e^{j2\pi\big(\frac{k\nu}{K}-\frac{l\mu}{L}\big)}\negthinspace e^{-j2\pi\big(\frac{k'\nu'}{K}-\frac{l'\mu'}{L}\big)}\nonumber\\
    &\negthinspace=\negthinspace\frac{1}{\negthinspace K L\negthinspace}\negthinspace\sum_{\nu=0}^{K-1}\negthinspace\sum_{\mu=0}^{L-1}\negthinspace P\negthinspace_xe^{j2\pi\frac{(k-k')\nu}{K}}\negthinspace e^{-j2\pi\frac{(l-l')\mu}{L}}\negthinspace=\negthinspace P_x\delta[k\negthinspace-\negthinspace k']\delta[l\negthinspace-\negthinspace l'].\negthinspace
\end{align}
Consequently, although the elements of $\mathbf{X}$ are not i.i.d., they are uncorrelated with zero mean and identical variance.
%The following derivation process is the same as OFDM. 
Therefore, when $K,L\gg1$ and the elements in $\mathbf{x}$ are zero-mean i.i.d. or more generally, second-order uncorrelated with zero mean and identical variance, the CRLB of OTFS has the same expression as OFDM in (\ref{eq.ofdmcrlb}) and (\ref{eq.crlbofdmnew}). 

It is worth noting that the CRLBs in this work are derived based on the signal characteristics, including RMS bandwidth and time, power spectrum, ESNR, etc., and therefore reflect the fundamental limitation determined by the waveform itself. 
In contrast, the study on OTFS \cite{8757044} explicitly accounts for the impact of \ac{ICI} and \ac{ISI} on OTFS, which leads to different expressions compared to OFDM, although the simulation results show that their CRLBs are almost identical.
From our perspective, effects such as ICI and ISI are relevant to signal processing schemes and can be suppressed by various approaches \cite{11151207,10319769,9925198,9299695}, hence the CRLB is expected to characterize the fundamental accuracy limitation of the waveform, rather than implementation-dependent performance.
Since the OTFS signal is obtained by ISFFT of the DD domain symbols $\mathbf{x}$ followed by OFDM modulation, the resulting time-frequency domain symbols $\mathbf{X}$ have the same second-order statistics as those of OFDM, resulting in the same CRLB expression under identical resource allocation. This conclusion is consistent with the simulation result in \cite{8757044}.

%hence the condition $\mathbb{E}\{X_{k,l}X_{k',l'}^*\}=P_X\delta(k-k')\delta(l-l')$ is equivalent to
%\begin{align}
%    \mathbb{E}\{\mathbf{X}\mathbf{X}^H\}=\mathbb{E}\{\mathbf{F}_L\mathbf{x}\mathbf{F}_K^H\mathbf{F}_K\mathbf{x}^H\mathbf{F}_L^H\}=P_x,
%\end{align}
%which holds since the FFT and IFFT matrices are unitary. 

In summary, this section provides the derivation of CRLBs for different waveforms based on the generic signal model, and it is demonstrated that $C_0$ and $\mathcal{I}\{C_1\}$ can be (approximately) eliminated. As a result, the matrix of CRLB is written as
\begin{align}
    \mathbf{C}=\mathbf{E}^{-1}=\mathrm{diag}([C_\tau,C_{f_\text{D}},C_{\theta_\text{R}}]).
\end{align}
This form excludes the impact of coupling elements and is more tractable for transformation to state parameters and optimization designs (RRA, beamforming, etc.).

\section{Impact of Virtual Array}\label{sec.VA}
VA is used in collocated MIMO radar for spatial resolution improvement and requires extending the inter-element spacing of either the Tx or the Rx antenna array. Previous studies \cite{li2021thinned,7300572,6638411} primarily focus on the improved AoA CRLB, while the impact of different multiplexing schemes on delay and Doppler CRLBs and the changes in ESNR due to the non-coherent transmission are missing.

In this work, we assume the space extension at the Rx, i.e., $d_\text{R}=N_\text{T}\lambda/2=N_\text{T}d_\text{T}$, whereas the conclusion is the same as that with the extension at the Tx. 
We use $(n_\text{R},n_\text{T})$ to denote the element corresponding to the $n_\text{R}$-th Rx antenna and the $n_\text{T}$-th Tx antenna, and its index in the virtual array is $n_\text{v}=n_\text{T}+(n_\text{R}-1)N_\text{T}$, whose size is $N_\text{v}=N_\text{T}N_\text{R}$.
Without loss of generality, we further assume that the Tx antenna elements have the same transmission power, yielding $P_{n_\text{T}}=P_\text{s}$ and $E_{n_\text{T}}=E_\text{v}=\int |s_{n_\text{T}}(t)|^2dt,\forall n_\text{T}$.
VA requires multiplexing and orthogonal transmission signals:
\begin{align}
    \int s_{n_\text{T}}(t)s_{n_\text{T}'}^*(t)dt=E_\text{v}\delta[n_\text{T}-n_\text{T}'].\label{eq.othosig}
\end{align}

The received signal $\mathbf{r}(t)\in\mathbb{C}^{N_\text{v}\times1}$ can be expressed as
\begin{align}\label{eq.vart}
    \mathbf{r}(t)=ae^{j\phi}\mathbf{a}_{\text{R}}^*(\theta)\otimes(\mathbf{a}_{\text{T}}^*(\theta)\odot\mathbf{s}(t-\tau)e^{j2\pi f_\text{D}t})+\mathbf{n}(t),
\end{align}
where $\mathbf{s}(t)=[s_0(t),s_1(t),...,s_{N_\text{T}-1}(t)]^T$ implies multiplexed transmitter signal. $\theta$ equals AoD and AoA in monostatic radar. The orthogonal signals result in a non-coherent transmission scheme that cannot benefit from coherent gain of beamforming, hence the amplitude without beamforming $a=A/N_\text{T}$ defined in Section~\ref{sec.model} instead of $A$ is considered here.
(\ref{eq.vart}) shows a relatively untractable form than (\ref{eq.rec}). However, the received signal corresponding to the $n_\text{T}$-th Tx antenna element can be simply expressed as 
\begin{align}
    \mathbf{r}_{n_\text{T}}(t)=ae^{j\phi}\mathbf{a}_{\text{R}}^*(\theta){a}_{n_\text{T}}^*(\theta)s_{n_\text{T}}(t-\tau)e^{j2\pi f_\text{D}t}+\mathbf{n}_{n_\text{T}}(t),\label{eq.r_nt}
\end{align}
where $a_{n_\text{T}}(\theta)$ denotes the $n_\text{T}$-th element in $\mathbf{a}_\text{T}(\theta)$. (\ref{eq.r_nt}) is in the same form as (\ref{eq.rec}), hence its FIM, omitting the decoupled amplitude, can be written as
\begin{align}
    \mathbf{F}_{n_\text{T}}=\begin{bmatrix}
        F_{\phi\phi,n_\text{T}} & 0 & F_{\phi f_\text{D},n_\text{T}} & F_{\phi\theta,n_\text{T}}\\
        0 & F_{\tau\tau,n_\text{T}} & 0 & 0\\
        F_{\phi f_\text{D},n_\text{T}} & 0 & F_{f_\text{D}f_\text{D},n_\text{T}} & F_{f_\text{D}\theta_\text{R},n_\text{T}}\\
        F_{\phi\theta,n_\text{T}} & 0 & F_{f_\text{D}\theta_\text{R},n_\text{T}} & F_{\theta_\text{R}\theta_\text{R},n_\text{T}}
    \end{bmatrix},
\end{align}
with the condition of $C_0=\mathcal{I}\{C_1\}=0$.

\begin{proposition}\label{p.2}
    In a radar system with VA, the global FIM can be calculated by the sum of the local FIMs corresponding to each transmitter antenna elements, i.e., $\mathbf{F}^\text{v}=\sum_{n_\text{T}=0}^{N_\text{T}-1}\mathbf{F}_{n_\text{T}}$.
\end{proposition}
\begin{proof}
    Given the parameter $\pmb{\Theta}$, if the observations of different $\mathbf{r}_{n_\text{T}}(t),\forall {n}_\text{T}\in[0,N_\text{T}-1]$ are independent of each other, the log-likelihood function can be written as the sum of individual log-likelihood terms, i.e., 
    \begin{align}\label{eq.llhsum}
        &\log p(\mathbf{r}|\pmb{\Theta})=-\frac{1}{\sigma^2}\int||\mathbf{r}(t)-\pmb{\mu}(t)||^2dt\nonumber\\
        &=-\frac{1}{\sigma^2}\sum_{n_\text{T}=0}^{N_\text{T}-1}\int||\mathbf{r}_{n_\text{T}}(t)-\pmb{\mu}_{n_\text{T}}(t)||^2dt=\sum_{n_\text{T}=0}^{N_\text{T}-1}\log p(\mathbf{r}_{n_\text{T}}|\pmb{\Theta}).
    \end{align}
    This requires that the noise terms $\mathbf{n}_{n_\text{T}}(t)\sim\mathcal{CN}(0,\sigma^2)^{N_\text{R}\times1}$ are independent of each other. Note that for Gaussian variables, uncorrelated means independent. %The uncorrelated property can be written as
    %\begin{align}\label{eq.othnoise}
    %    &\mathbb{E}\{\mathbf{n}_{n_\text{T}}(t)\mathbf{n}_{n_\text{T}'}^H(t)\}=\sigma^2\mathbf{I}_{N_\text{R}}\delta(n_\text{T}-n_\text{T}').
    %\end{align}
    
    For TDM and FDM, the AWGN terms $\mathbf{n}_{n_\text{T}}(t)$ with different $n_\text{T}$ are orthogonal in time or frequency, thus they are independent of each other.
    
    For CDM, although the signals are not separated in frequency or time, they are orthogonal in correlation due to the orthogonality of the applied codes, such as Hadamard code.
    After outer decoding, the noise terms are uncorrelated to each other since
    \begin{align}
        &\mathbb{E}\{\mathbf{n}_{n_\text{T}}(t)\mathbf{n}_{n_\text{T}'}^H(t)\}=\frac{1}{T_\text{F}}\int_{-T_\text{F}/2}^{T_\text{F}/2} c_{n_\text{T}}(t)c_{n_\text{T}'}^*(t)\mathbf{n}(t)\mathbf{n}^H(t)dt\nonumber\\
        &=\sigma^2\mathbf{I}_{N_\text{R}}\delta[n_\text{T}- n_\text{T}'],\label{eq.othnoi}
    \end{align}
    where $c_{n_\text{T}}(t)$ denotes the outer coding vector for the $n_\text{T}$-th signal, which is orthogonal to each other, i.e., $\int c_{n_\text{T}}(t)c^*_{n_\text{T}'}(t)dt= T_\text{F}\delta[n_\text{T}-n_\text{T}']$.
    In summary, (\ref{eq.llhsum}) holds for VA with different multiplexing schemes.

    Consequently, the FIM can be written as the sum of the individual FIMs:
    \begin{align}
        &\mathbf{F}_{ij}^\text{v}=\mathbb{E}\Big\{\frac{\partial \log p(\mathbf{r}|\pmb{\Theta})}{\partial\pmb{\Theta}_i}\frac{\partial \log p(\mathbf{r}|\pmb{\Theta})}{\partial\pmb{\Theta}_j}\Big\}\nonumber\\
        &=\sum_{n_\text{T}=0}^{N_\text{T}-1}\sum_{n_\text{T}'=0}^{N_\text{T}-1}\mathbb{E}\Big\{\frac{\partial \log p(\mathbf{r}_{n_\text{T}}|\pmb{\Theta})}{\partial\pmb{\Theta}_i}\frac{\partial \log p(\mathbf{r}_{n_\text{T}'}|\pmb{\Theta})}{\partial\pmb{\Theta}_j}\Big\}\nonumber\\
        &=\sum_{n_\text{T}=0}^{N_\text{T}-1}\mathbb{E}\Big\{\frac{\partial \log p(\mathbf{r}_{n_\text{T}}|\pmb{\Theta})}{\partial\pmb{\Theta}_i}\frac{\partial \log p(\mathbf{r}_{n_\text{T}}|\pmb{\Theta})}{\partial\pmb{\Theta}_j}\Big\}\nonumber\\
        &\quad+\sum_{n_\text{T}=0}^{N_\text{T}-1}\sum_{\substack{n_\text{T}'=0,\\ n_\text{T}'\ne n_\text{T}}}^{N_\text{T}-1}\mathbb{E}\Big\{\frac{\partial \log p(\mathbf{r}_{n_\text{T}}|\pmb{\Theta})}{\partial\pmb{\Theta}_i}\frac{\partial \log p(\mathbf{r}_{n_\text{T}'}|\pmb{\Theta})}{\partial\pmb{\Theta}_j}\Big\}\nonumber\\
        &=\sum_{n_\text{T}=0}^{N_\text{T}-1}\mathbf{F}_{n_\text{T}},
    \end{align}
    where the cross terms are eliminated since when $n_\text{T}\ne n_\text{T}'$, $\frac{\partial \log p(\mathbf{r}_{n_\text{T}}|\pmb{\Theta})}{\partial\pmb{\Theta}_i}$ and $\frac{\partial \log p(\mathbf{r}_{n_\text{T}'}|\pmb{\Theta})}{\partial\pmb{\Theta}_j}$ are with zero mean and uncorrelated.
\end{proof}

%For TDM and FDM, the signals $s_{n_\text{T}}$ with different $n_\text{T}$ are orthogonal in time or frequency, the observations of them are statistically independent, i.e., $p(\mathbf{r}|\pmb{\Theta})=\prod_{n_\text{T}=0}^{N_\text{T}-1}p(\mathbf{r}_{n_\text{T}}|\pmb{\Theta})$. 
%The global FIM can be calculated by the sum of $N_\text{R}$ local FIMs since the multiplexing guarantees the orthogonality between the transmission signals, i.e., $\mathbf{F}=\sum_{n_\text{T}=0}^{N_\text{T}-1}\mathbf{F}_{n_\text{T}}$, %Since $F_{\tau\tau,n_\text{T}}$ is decoupled in $\mathbf{F}_{n_\text{T}}$, we build $\mathbf{F}'_{n_\text{T}}$ by excluding the rows and columns of delay, 
Consequently, the elements of the resulting $\mathbf{F}^\text{v}$ is given by
\begin{align}
    \negthinspace &F_{\tau\tau}^\text{v}\negthinspace=\negthinspace\frac{2}{\sigma^2}\negthinspace a^2\negthinspace N_\text{R}\negthinspace\sum_{n_\text{T}=0}^{N_\text{T}-1}\negthinspace\int \negthinspace|\dot{s}_{n_\text{T}}(t\negthinspace-\negthinspace\tau)|^2dt\negthinspace\nonumber\\
    &=\frac{8\pi^2\negthinspace a^2 N_\text{R}}{\sigma^2}\sum_{n_\text{T}=0}^{N_\text{T}-1}{\int f^2|S_{n_\text{T}}(f)|^2df}=\negthinspace8\pi^2\negthinspace N_\text{R}\gamma_\text{v} \negthinspace\sum_{n_\text{T}=0}^{N_\text{T}-1}\negthinspace B_{\text{rms},n_\text{T}}^2\negthinspace.\negthinspace\label{eq.fttv}
\end{align}
\begin{align}\label{eq.fphiphiv}
     F_{\phi\phi}^\text{v}=\frac{2}{\sigma^2} N\negthinspace_\text{R}a ^2\sum_{n_\text{T}=0}^{N_\text{T}-1}\int  |{s}_{n_\text{T}}( t-\tau)|^2dt=2N_\text{v}\gamma_{\text{v}}.
\end{align}
\begin{align}\label{eq.fphifv}
     &F_{\phi f_\text{D}}^\text{v}=\frac{4}{\sigma^2}\pi a^2N_\text{R} \sum_{n_\text{T}=0}^{N_\text{T}-1}\int (t+T_0) |s_{n_\text{T}}(t-\tau)|^2 dt\nonumber\\
    &\overset{u=t-\tau}{=}%\negthinspace\frac{4}{\sigma^2}\pi a^2N_\text{R}\Big(\negthinspace(T_0+\tau) N_\text{T}E_{\text{v}}\negthinspace+\negthinspace\sum_{n_\text{T}=0}^{N_\text{T}-1}\negthinspace\int u|s_{n_\text{T}}(u)|^2du\negthinspace\Big)\nonumber\\
    4\pi N_\text{v}(T_0+\tau)\gamma_\text{v}.
\end{align}
\begin{align}
    &F_{\phi\theta}^\text{v}=-\frac{2}{\sigma^2}\mathcal{R}\Big\{ ja^2\sum_{n_\text{T}=0}^{N_\text{T}-1} (\mathbf{a}_\text{R}^T(\theta)\dot{\mathbf{a}}^*_\text{R}(\theta) +N_\text{R}a_{n_\text{T}}(\theta)\dot{a}^*_{n_\text{T}}(\theta))\nonumber\\&\qquad \cdot\int_{-\infty}^{\infty}|s_{n_\text{T}}(t-\tau)|^2dt\Big\}\nonumber\\
    &=\negthinspace-\frac{a^2\pi\negthinspace\cos(\theta)(\negthinspace N\negthinspace_\text{v}\negthinspace-\negthinspace1\negthinspace)N\negthinspace_\text{v}E_{\text{v}}}{\sigma^2}\negthinspace=\negthinspace-\pi N\negthinspace_\text{v}(\negthinspace N\negthinspace_\text{v}\negthinspace-\negthinspace1\negthinspace)\negthinspace\cos(\theta)\gamma_\text{v}.\negthinspace
    \label{eq.phitheta}
\end{align}
\begin{align}
    &F_{f_\text{D}f_\text{D}}^\text{v}=\frac{8\pi^2a^2N_\text{R}}{\sigma^2}\sum_{n_\text{T}=0}^{N_\text{T}-1}\int (t+T_0)^2|s_{n_\text{T}}(t-\tau)|^2dt\nonumber\\
    &=\frac{8\pi^2 a^2 N_\text{R}}{\sigma^2}\sum_{n_\text{T}=0}^{N_\text{T}-1}\int (u^2+(T_0+\tau)^2)|s_{n_\text{T}}(u)|^2 du\nonumber\\
    &={8\pi^2  N_\text{R}}\gamma_\text{v}\Big((T_0+\tau)^2N_\text{T}+\sum_{n_\text{T}=0}^{N_\text{T}-1}T_{\text{rms},n_\text{T}}^2\Big).\label{eq.fffv}
\end{align}
\begin{align}
    & F_{f_\text{D}\theta}^\text{v}=-\frac{2}{\sigma^2}\mathcal{R}\Big\{j2\pi a^2(T_0+\tau)E_{\text{v}}\nonumber\\
    &\qquad\cdot\sum_{n_\text{T}=0}^{N_\text{T}-1}(\mathbf{a}_\text{R}^T(\theta)\dot{\mathbf{a}}^*_\text{R}(\theta) +N_\text{R}a_{n_\text{T}}(\theta)\dot{a}^*_{n_\text{T}}(\theta))\Big\}\nonumber\\
    &=-{2}\pi^2(T_0+\tau)N_\text{v}(N_\text{v}-1)\cos(\theta)\gamma_{\text{v}}.\label{eq.ftheta}
\end{align}
\begin{align}
    &F_{\theta\theta}^\text{v}
    \negthinspace=\negthinspace\frac{2}{\sigma^2}a^2E_{\text{v}}\mathcal{R}\Big\{\negthinspace\sum_{n_\text{T}=0}^{N_\text{T}-1}\negthinspace||(\dot{\mathbf{a}}_\text{R}(\theta)a_{n_\text{T}}(\theta)+{\mathbf{a}}_\text{R}(\theta)\dot{a}_{n_\text{T}}(\theta))||^2\negthinspace\Big\}\nonumber
    \\
    &=\frac{\pi^2\cos^2(\theta)\gamma_\text{v}(N_\text{v}-1)N_\text{v}(2N_\text{v}-1)}{3},\label{eq.thetatheta}
\end{align}
where $\gamma_\text{v}=a^2E_{\text{v}}/\sigma^2$ is the ESNR at each virtual antenna element. With time centralization, $\sum_{n_\text{T}=0}^{N_\text{T}-1}\int t|s_{n_\text{T}}(t)|^2dt=0$ still holds since for TDM, $\sum_{n_\text{T}=0}^{N_\text{T}-1}\int t|s_{n_\text{T}}(t)|^2dt=\int t|s(t)|^2dt$; for FDM and CDM, $\int t|s_{n_\text{T}}(t)|^2dt=0,\forall n_\text{T}$. Hence, the corresponding terms in $F^\text{v}_{\phi f_\text{D}},F^\text{v}_{f_\text{D}f_\text{D}},F^\text{v}_{f_\text{D}\theta}$ are not illustrated.

\subsection{TDM}
TDM guarantees the signals are orthogonal in the time domain.
\Ac{ITDM} is the most common TDM scheme in VA sensing. It allocates the transmission of different Tx antenna elements into time slots in a sequentially interleaved scheme. 
The $n_\text{T}$-th transmitted signal is given by
\begin{align}
    s_{n_\text{T}}(t)=\sum_{k=0}^{K-1}s(t)\text{rect}\Big(\frac{t-n_\text{T}T-kN_\text{T}T+T_0}{T}\Big).
\end{align}

Another TDM scheme is \ac{BTDM}, where the transmissions of different antennas are temporally compressed within a fixed sub-frame duration $T_\text{F}/N_\text{T}$, the $n_\text{T}$-th transmitted signal is given by
\begin{align}
    s_{n_\text{T}}(t)=s(t)\text{rect}\Big(\frac{t-n_\text{T}KT+T_0}{KT}\Big).
\end{align}

For TDM signals, the relationship between PRI $T$ and frame length $T_\text{F}$ is changed to $T_\text{F}=KN_\text{T}T$. 
Under the assumption that the average transmission power $P_\text{s}$ is not influenced, the signal energy and ESNR are reduced to $E_\text{v}=E_\text{s}/N_\text{T}$ and $\gamma_\text{v}=a^2E_\text{v}/\sigma^2=\gamma/N_\text{T}^3$ due to the reduced transmission time. The TDM signals are orthogonal and complementary in time, hence
\begin{align}
    &\sum_{n_\text{T}=0}^{N_\text{T}-1}\int |\dot{s}_{n_\text{T}}(t-\tau)|^2dt=\int |\dot{s}(t-\tau)|^2dt,\nonumber\\ &\sum_{n_\text{T}=0}^{N_\text{T}-1}\int |{s}_{n_\text{T}}(t-\tau)|^2dt=\int |{s}(t-\tau)|^2dt=E_\text{s},\nonumber\\
    &\sum_{n_\text{T}=0}^{N_\text{T}-1}\negthinspace T_{\negthinspace \text{rms},n_\text{T}}^2\negthinspace =\negthinspace \sum_{n_\text{T}=0}^{N_\text{T}-1}\negthinspace \frac{\negthinspace \negthinspace \int\negthinspace  t^2|{s}_{n_\text{T}}(t\negthinspace -\negthinspace \tau)|^2\negthinspace  dt}{\negthinspace \negthinspace \int\negthinspace  |{s}_{n_\text{T}}(t\negthinspace -\negthinspace \tau)|^2dt}\negthinspace =\negthinspace \frac{\negthinspace \int\negthinspace  t^2|s(t)|^2dt}{E_\text{v}}
    \negthinspace =\negthinspace  N_\text{T}T_{\text{rms}}^2.\negthinspace \label{eq.tcomort}
\end{align}
%The relationship between the elements of ITDM FIM $\mathbf{F}^\text{v}$ and that without VA is calculated by
%\begin{align}
%    & F_{\tau\tau}^\text{v}= F_{\tau\tau}/N_\text{T}^2,\, F_{\phi\phi}^\text{v}= F_{\phi\phi}/N_\text{T}^2,\, F_{\phi f}^\text{v}= F_{\phi f}/N_\text{T}^2,\nonumber\\
%    & F_{ f f}^\text{v}= F_{ f f}/N_\text{T}^2,
%\end{align}
Substituting $\gamma_\text{v}=\gamma/N_\text{T}^3$ and (\ref{eq.tcomort}) into \cref{eq.fttv,eq.fphiphiv,eq.fphifv,eq.phitheta,eq.fffv,eq.ftheta,eq.thetatheta}, the relationship between the CRLBs of TDM-VA and those without VA is derived as
\begin{align}
    &C_{\tau,\text{TDM}}=N_\text{T}^2{C_\tau},\ C_{f_\text{D},\text{TDM}}={N_\text{T}^2}{C_{f_\text{D}}},\nonumber\\
    &C_{\theta,\text{TDM}}=\frac{6}{\pi^2\cos^2(\theta)N_\text{v}(N_\text{v}^2-1)\gamma_\text{v}}\approx C_{\theta_\text{R}},\label{eq.tdmcrlb}
\end{align}
where the approximation holds for $N_\text{v}\gg1$. The degradation mainly comes from the non-coherent transmission scheme and reduced signal energy. While the delay and velocity CRLBs significantly increase, the result of AoA approximates that without VA. 

%In summary, for VA with TDM, if the multiplexed signals are complementary and orthogonal in time and the transmission power of each antenna element is equal and unchanged compared to the system without VA, the CRLB of delay, Doppler, and AoA can be expressed by (\ref{eq.tdmcrlb}), the reason is given in (\ref{eq.tcomort}). 

\subsection{FDM}
The signals are transmitted continuously and separated in the frequency domain, hence the time domain characteristics of $s_{n_\text{T}}(t)$ like $T_\text{rms}^2$ are the same as $s(t)$. For single-carrier waveforms like FMCW, the \ac{BFDM} is usually applied \cite{9627227}, where the signal spectrum of $s_{n_\text{T}}(t)$ is given by
\begin{align}
    S_{n_\text{T}}(f)=\sqrt{N_\text{T}}S(f)\text{rect}\Big(\frac{f-n_\text{T}{B/}{N_\text{T}}+B/2}{B/N_\text{T}}\Big),
\end{align}
where $\sqrt{N_\text{T}}$ exists under the assumption that the time domain power is still $P_\text{s}$. The signal energy and ESNR are then given by $E_\text{v}=E_\text{s}$ and $\gamma_\text{v}=\gamma/N_\text{T}^2$. 

For multi-carrier waveforms like OFDM and OTFS, the \ac{CFDM} with an interleaved subcarrier assignment scheme is usually applied to minimize the \ac{ICI}. The corresponding spectrum of $s_{n_\text{T}}(t)$ is given by 
\begin{align}
    S_{n_\text{T}}(f)\negthinspace=\negthinspace\sqrt{N_\text{T}}S(f)\sum_{l=0}^{L-1}\text{rect}\Big(\frac{f-n_\text{T}\Delta f-lN_\text{T}\Delta f+B/2}{\Delta f}\Big),
\end{align}
where $\Delta f=\frac{B}{LN_\text{T}}$ for FDM. In contrast to BFDM, this scheme avoids degradation of range resolution at the cost of a reduced maximum unambiguous range. 

The FDM signals are orthogonal and complementary in the frequency domain, hence
\begin{align}
    &\sum_{n_\text{T}=0}^{N_\text{T}-1}\int |{S}_{n_\text{T}}(f)|^2df= N_\text{T}\int |{S}(f)|^2df=N_\text{T}E_\text{s},\nonumber\\
    &\negthinspace \sum_{n_\text{T}=0}^{N_\text{T}-1}\negthinspace B_{\negthinspace\text{rms},n_\text{T}}^2\negthinspace=\negthinspace\sum_{n_\text{T}=0}^{N_\text{T}-1}\negthinspace\frac{\negthinspace \int\negthinspace f^2\negthinspace |\negthinspace  S_{n_\text{T}}\negthinspace (f)\negthinspace |\negthinspace ^2\negthinspace df}{\negthinspace \int\negthinspace  |S_{n_\text{T}}(f)|^2df}\negthinspace=\negthinspace\frac{N\negthinspace _\text{T}\negthinspace\int\negthinspace f^2\negthinspace |\negthinspace  S(f)\negthinspace |\negthinspace ^2\negthinspace df}{E_\text{s}}
    \negthinspace=\negthinspace N_\text{T}B_{\text{rms}}^2.\nonumber\\\negthinspace 
    &\negthinspace\sum_{n_\text{T}=0}^{N_\text{T}-1}\negthinspace T_{\text{rms},n_\text{T}}^2\negthinspace =\negthinspace\sum_{n_\text{T}=0}^{N_\text{T}-1}\negthinspace \frac{\negthinspace  \int\negthinspace  t^2|{s}_{n_\text{T}}( t\negthinspace -\negthinspace \tau)|^2\negthinspace  dt}{ \negthinspace \int\negthinspace  |{s}_{n_\text{T}}(t\negthinspace -\negthinspace \tau)|^2dt}\negthinspace %=\negthinspace \frac{\negthinspace N_\text{T}\negthinspace\int\negthinspace  t^2\negthinspace|s(t)|^2\negthinspace dt}{E_\text{s}}
    \negthinspace=\negthinspace\sum_{n_\text{T}=0}^{N_\text{T}-1}\negthinspace T_{\text{rms},n_\text{T}}^2 \negthinspace=\negthinspace  N\negthinspace  _\text{T}T_{\text{rms}}^2.\negthinspace\label{eq.fcomort}
\end{align}
Substituting $\gamma_\text{v}=\gamma/N_\text{T}^2$ and (\ref{eq.fcomort}) into \cref{eq.fttv,eq.fphiphiv,eq.fphifv,eq.phitheta,eq.fffv,eq.ftheta,eq.thetatheta}, the relationship between the CRLBs of FDM-VA and those without VA is derived as
\begin{align}
    &C_{\tau,\text{FDM}}=N_\text{T}{C_\tau},\ C_{f_\text{D},\text{FDM}}={N_\text{T}}{C_{f_\text{D}}},\nonumber\\
    &C_{\theta,\text{FDM}}=\frac{6}{\pi^2\cos^2(\theta)N_\text{v}(N_\text{v}^2-1)\gamma_\text{v}}\approx\frac{1}{N_\text{T}}C_{\theta_\text{R}},\label{eq.tdmcrlb}
\end{align}
where the improvement with a factor of $N_\text{T}$ compared to TDM comes from the continuous transmission.

\subsection{CDM}
CDM is primarily applied in PMCW-based systems. 
Different from TDM and FDM, the signals $s_{n_\text{T}}(t)$ occupy the whole time-frequency plane. The $n_\text{T}$-th transmitted signal is given by
\begin{align}
    s_{n_\text{T}}(t)=c_{n_\text{T}}(t)s(t).
\end{align}
Note the expression of the frame length is changed to $T_\text{F}=KN_\text{T}\beta T$, where $\beta$ is the repetition factor for avoiding \ac{ISI} \cite{11456233} in MIMO PMCW radar systems, which is implemented by discarding the first of each $\beta$ PRIs at the Rx, hence the ESNR is reduced to $\gamma_\text{v}=\frac{(\beta-1)\gamma}{\beta N_\text{T}^2}$.
$c_{n_\text{T}}(t)=\sum_{k=0}^{K-1}\sum_{i=0}^{N_\text{T}-1}c_{n_\text{T}}^i\text{rect}\Big(\frac{t-i \beta T-kN_\text{T}\beta T}{ \beta T}\Big)$ with $c_{n_\text{T}}^i=\mathbf{H}_\text{d}(n_\text{T},i)$ represents the outer coding signal based on the Hadamard matrix $\mathbf{H}_\text{d}$. %The outer coding signals are orthogonal to each other, i.e., $\int c_{n_\text{T}}(t)c^*_{n_\text{T}'}(t)dt= T_\text{F}\delta[n_\text{T}-n_\text{T}']$, hence the orthogonality in (\ref{eq.othosig}) and (\ref{eq.othnoi}) also holds for CDM signals.

The FIM-relevant characteristics for CDM are listed as follows:
\begin{align}
    &\negthinspace\sum_{n_\text{T}=0}^{N_\text{T}-1}\negthinspace\int\negthinspace |\dot{s}_{n_\text{T}}(t\negthinspace-\negthinspace\tau)|^2\negthinspace dt\negthinspace=\negthinspace\sum_{n_\text{T}=0}^{N_\text{T}-1}\negthinspace\int\negthinspace |c_{n_\text{T}}(t)\dot{s}(t)|^2\negthinspace dt\negthinspace=\negthinspace N_\text{T}\negthinspace\int\negthinspace |\dot{s}(t)|^2\negthinspace dt,\negthinspace\nonumber\\
    &\negthinspace\sum_{n_\text{T}=0}^{N_\text{T}-1}\int |{s}_{n_\text{T}}(t-\tau)|^2dt=N_\text{T}\int |{s}(t-\tau)|^2dt=N_\text{T}E_\text{s},\nonumber\\
    &\negthinspace\sum_{n_\text{T}=0}^{N_\text{T}-1}\negthinspace T_{\text{rms},n_\text{T}}^2\negthinspace =\negthinspace\sum_{n_\text{T}=0}^{N_\text{T}-1}\negthinspace \frac{\negthinspace \negthinspace \int\negthinspace  t^2\negthinspace|\negthinspace{s}\negthinspace_{n_\text{T}}\negthinspace(t\negthinspace -\negthinspace \tau)\negthinspace|^2\negthinspace  dt}{\negthinspace \negthinspace \int\negthinspace  |{s}_{n_\text{T}}\negthinspace(t\negthinspace -\negthinspace \tau)|^2dt}\negthinspace =\negthinspace \frac{\negthinspace N_\text{T}\negthinspace\int\negthinspace  t^2\negthinspace|\negthinspace s(t)\negthinspace|^2\negthinspace dt}{E_\text{s}}
    \negthinspace =\negthinspace  N_\text{T}T_{\text{rms}}^2,\negthinspace\label{eq.cdmcomp}
\end{align}
%with the approximation of $\dot{c}_{n_\text{T}}(t)\approx0$ due to the rectangular shape of ${c}(t)$ with period $\beta T$ much longer than the time scale of variation of ${s}_{n_\text{T}}(t)$. 
Substituting $\gamma_\text{v}=\frac{(\beta-1)\gamma}{\beta N_\text{T}^2}$ and (\ref{eq.cdmcomp}) into \cref{eq.fttv,eq.fphiphiv,eq.fphifv,eq.phitheta,eq.fffv,eq.ftheta,eq.thetatheta}, the relationship between the CRLBs of CDM-VA and those without VA is derived as
\begin{align}
    &C_{\tau,\text{CDM}}=\frac{\beta}{\beta-1}N_\text{T}{C_\tau},\ C_{f_\text{D},\text{CDM}}=\frac{\beta}{\beta-1}{N_\text{T}}{C_{f_\text{D}}},\nonumber\\
    &C_{\theta,\text{CDM}}\negthinspace=\negthinspace\frac{6\beta}{(\beta-1)\pi^2\cos^2(\theta)N_\text{v}(N_\text{v}^2-1)\gamma_\text{v}}\negthinspace\approx\negthinspace\frac{\beta}{(\beta-1)N_\text{T}}C_{\theta_\text{R}}.\label{eq.tdmcrlb}
\end{align}
%where the terms with $\beta$ are produced by discarding the first of each $\beta$ PRIs at the Rx.

\textit{Remark:} Although the deployment of VA increases the delay and Doppler CRLB due to the non-coherent transmission scheme, it is worth noting that its advantages primarily lie in improved spatial resolution and target separability, especially in multi-target scenarios. 
%Furthermore, VA can achieve full-field perception without prior knowledge of the direction, providing greater flexibility and robustness than beamforming. 

\section{Numerical Results}\label{sec.simulation}
Simulations in this work are designed to show the CRLB behavior of the waveforms and VA sensing. 
Specifically, the impact of linear approximation of FMCW in (\ref{eq.crlbfmcw}), the performance of PMCW with various shaped pulses in (\ref{eq.pmcwf}), and the difference between continuous signal-based (\ref{eq.ofdmcrlb}) and discrete symbols-based (\ref{eq.crlbofdmnew}) OFDM/OTFS CRLBs are investigated. In addition, the impact of VA with various multiplexing schemes is evaluated. 

The default values of ESNR, bandwidth, and frame length are set to $10$\,dB, 400\,MHz, 10\,ms, corresponding to the configuration of 5G network operating in FR2. The carrier frequency is set to $28$\,GHz, and the Tx and Rx antennas are ULAs with $N_\text{T}=N_\text{R}=8$. Since the CRLB of AoA is irrelevant to the waveform, the CRLBs of range and radial velocity are observed, which are calculated by $C_r=(\frac{c_0}{2})^2C_\tau$ and $C_v=(\frac{c_0}{2f_\text{c}})^2C_{f_\text{D}}$ for monostatic radar systems, where $c_0$ denotes the speed of light in vacuum. 

%Table~\ref{tab:parameter} gives the default parameters of four waveforms. An automotive radar system is considered, where the carrier frequency is 77\,GHz \cite{9318758} with a bandwidth of 1\,GHz. The number of samples per PRI for FMCW, OFDM, and OTFS are calculated by $L=BT$

\begin{comment}
\begin{table}
    \centering
    \caption{System default parameters.}
    \begin{tabular}{c|c|c|c|c}
    \hline
    Parameters & FMCW & PMCW & OFDM & OTFS\\\hline
        Carrier frequency $f_\text{c}$ & \multicolumn{4}{c}{77\,GHz} \\\hline
        Bandwidth $B$ & \multicolumn{4}{c}{1\,GHz} \\\hline
        Samples per PRI  $L$ & 1024 & 1023 & \multicolumn{2}{c}{2048}\\\hline
        Number of PRIs $K$ & \multicolumn{4}{c}{1024}\\\hline
        PRI $T$ & 2.048\,$\mu$s & 2.046\,$\mu$s  & \multicolumn{2}{c}{2.048\,$\mu$s}\\\hline
        Subcarrier $\Delta f$ & \multicolumn{2}{c|}{-} & \multicolumn{2}{c}{488.28\,kHz}\\\hline
        Frame length $T_\text{F}$ & 2.097\,ms  & 2.095\,ms & \multicolumn{2}{c}{2.097\,ms}\\\hline
        \text{ESNR} $\gamma$ & \multicolumn{4}{c}{10\,dB}\\\hline
    \end{tabular}
    \label{tab:parameter}
\end{table}
\end{comment}

\subsection{Analysis of Waveforms}
The linear asymptotic of FMCW CRLB with increasing $K$ is illustrated in Fig.~\ref{fig:fmcwcrlb}, where the range and velocity CRLBs in (\ref{eq.crlbfmcw}) with and without the approximation of $1-\frac{1}{K^2}\approx1$ are illustrated. The accurately calculated CRLBs converge to the approximated ones with increased $K$, where the difference between them almost vanishes when $K\ge20$, which is generally satisfied in practical systems due to the requirement of velocity estimation capability. 

\begin{figure}
    \centering
    \includegraphics[width=\linewidth]{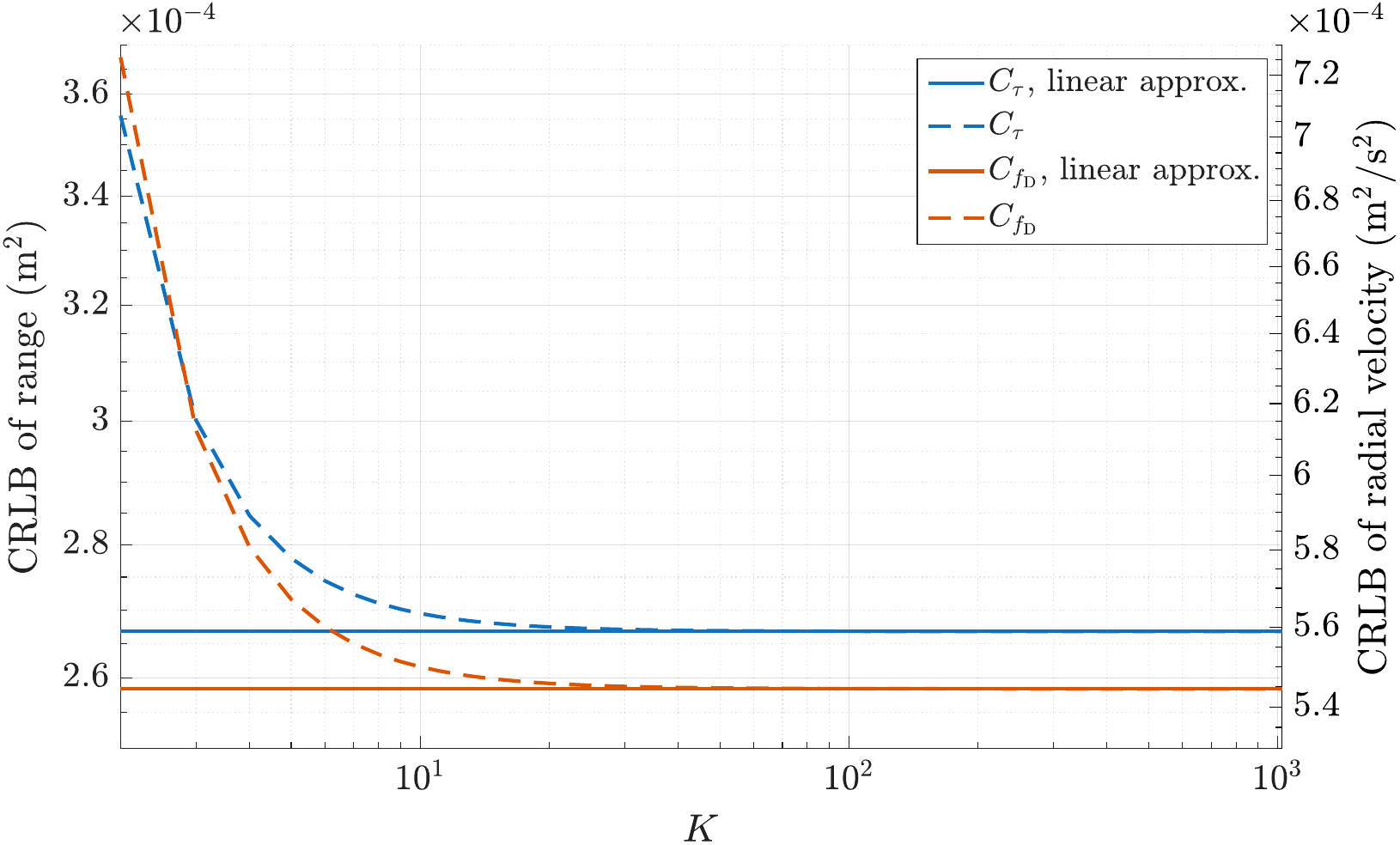}
    \caption{Accurate and approximated CRLBs of FMCW.}
    \label{fig:fmcwcrlb}
\end{figure}
In the investigation of PMCW, the emphasis is put on the behavior of shaped pulses, particularly for RC and RRC with varied roll-off factors, thus the CRLB of radial velocity with the identical expression for all pulses is not tested. Fig.~\ref{fig:pmcwcrlb} shows the result, where the result of the rectangular pulse is drawn as a constant line since it has no roll-off factor. When $\alpha$ approaches 0, RC and RRC converge to the sinc pulse, thus resulting in the same lowest CRLB. The CRLB increases with the growth of $\alpha$. While RRC always behaves better than the rectangular pulse, the performance of RC becomes even worse than it when $\alpha>0.6$. However, it is worth noting that the RC pulse is rarely implemented directly. It is usually used to express the effective response of the cascaded Tx and Rx RRC filters. In practical systems, RRC, rectangular, and truncated sinc pulses are commonly used.

\begin{figure}
    \centering
    \includegraphics[width=\linewidth]{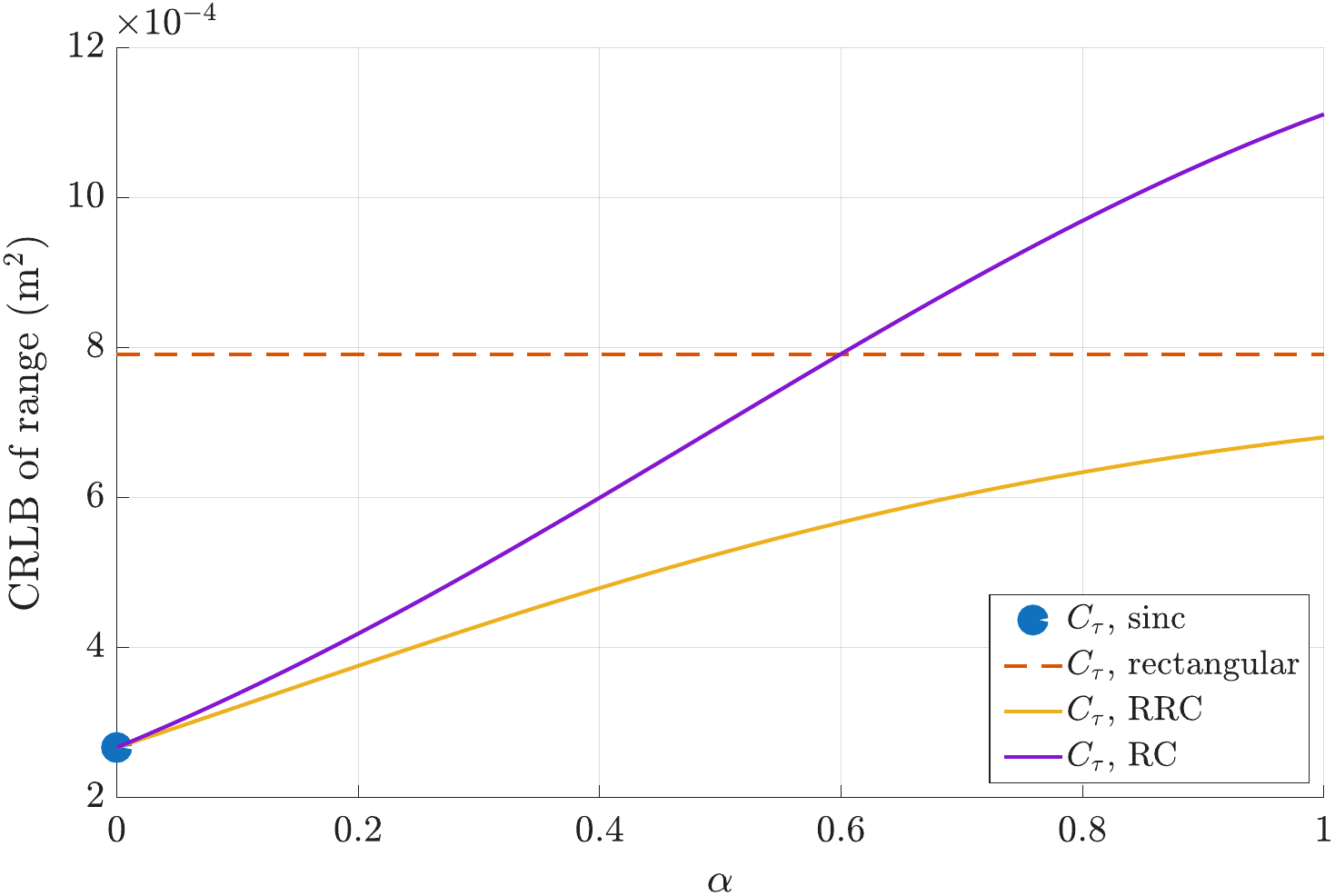}
    \caption{CRLBs of PMCW with different shaped pulses and roll-off factors.}
    \label{fig:pmcwcrlb}
\end{figure}

As demonstrated in \cite{8757044} and analyzed in Section~\ref{sec.otfs}, the gap in CRLB between OFDM and OTFS is negligible, hence the simulation of OFDM and OTFS focuses on the comparison between the generic continuous signal-based CRLB derived in this work and the discrete symbols-based one provided by waveform-specific studies, their ratios $C_r/C_r'$ and $C_{v}/C_{v}'$ respectively with varying $L$ and $K$ are shown in Fig.~\ref{fig:OFDM_lk}. The ratio approaches 0.99 when $L$ and $K$ grow to 10, and exceeds 0.9999 when $K,L>100$, demonstrating the consistency between generic signal-based and waveform-specific analyses. ISAC systems generally require high $K$ and $L$. For instance, the number of subcarriers and OFDM symbols in a 5G OFDM frame is on the order of $10^3$ and $10^2\sim10^3$, respectively, thus the gap between $C_r$ and $C_r'$, $C_{v}$ and $C_{v}'$ can be safely neglected. The result reflects the consistency between the proposed approach and the waveform-specific CRLB analysis, verifying the correctness and applicability of the unified framework.

\begin{figure}
    \centering
    \includegraphics[width=\linewidth]{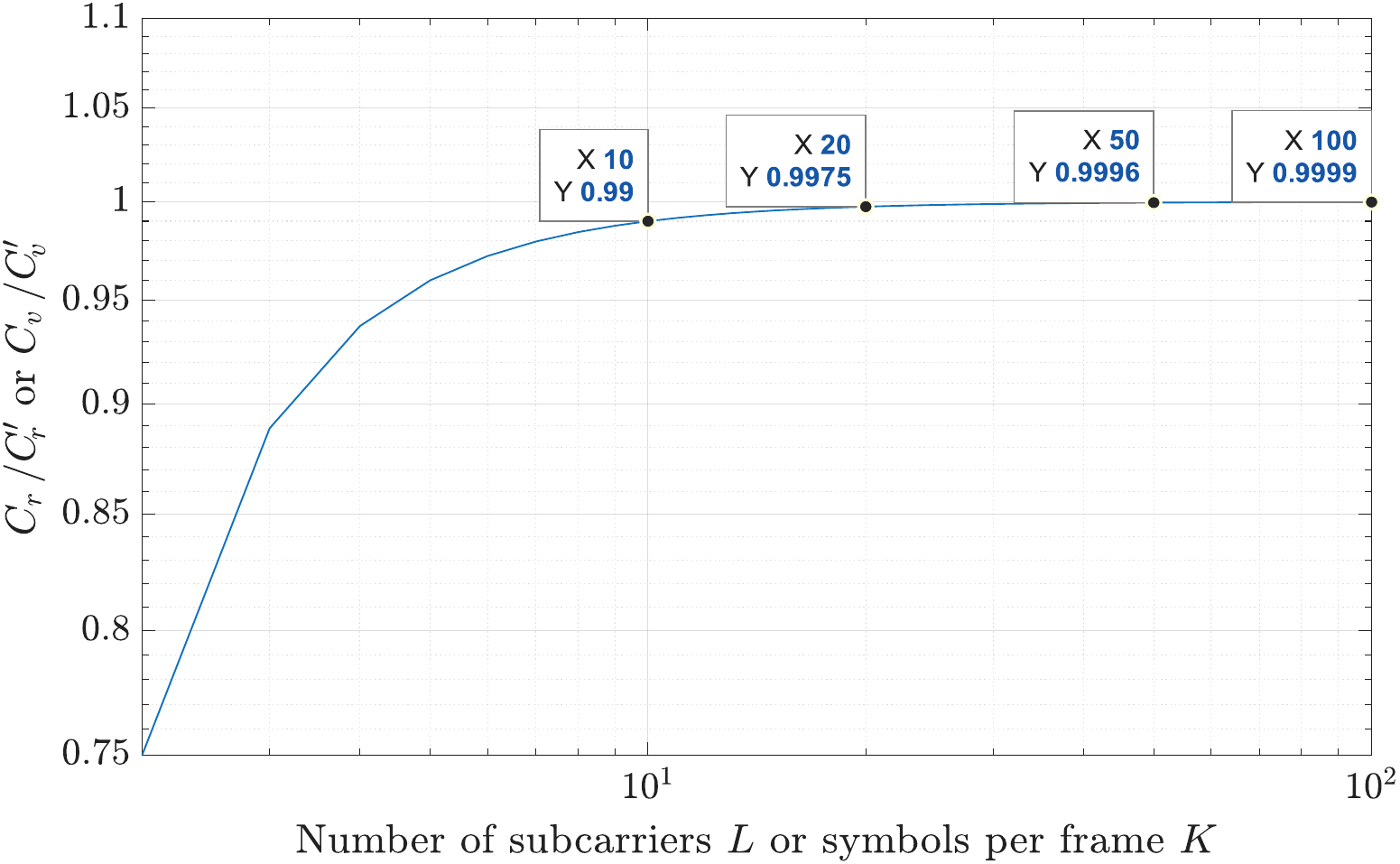}
    \caption{OFDM/OTFS: $C_r/C_r'$ and $C_{v}/C_{v}'$ with respectively varying $L$ and $K$.}
    \label{fig:OFDM_lk}
\end{figure}

\subsection{Impact of VA and Multiplexing}
The CRLBs with VA and different multiplexing schemes are compared in this section. Without loss of generality, the average transmission power of each antenna is assumed to be identical, and the ratio between the CRLBs with different multiplexing schemes and the CRLB without VA is observed. The system parameters follow the configuration defined at the beginning of Section~\ref{sec.simulation}, while $N_\text{T}$ is chosen as the variable since it determines the number of transmission/multiplexed signals and is the dominant factor influencing the CRLBs with VA. The result is shown in Fig.~\ref{fig.VA}, where the ratios in range and radial velocity CRLBs are identical, hence they are plotted only once. The deployment of VA leads to increased CRLB in range and radial velocity due to the non-coherent transmission. The CRLB of TDM has the most significant increase since $C_{\tau,\text{TDM}}/C_{\tau},C_{f_\text{D},\text{TDM}}/C_{f_\text{D}}= N_\text{T}^2$, while the results of FDM and CDM have the same slope of $N_\text{T}$, with a gap determined by $\beta$. The gap becomes negligible when $\beta\ge8$, however, a larger $\beta$ leads to higher sensitivity to dynamic environments \cite{10444556}, resulting in a trade-off between these two functionalities.

\begin{figure}
\centering
\subfloat[Range and radial velocity CRLB ratios.]{
    \includegraphics[width=0.95\linewidth]{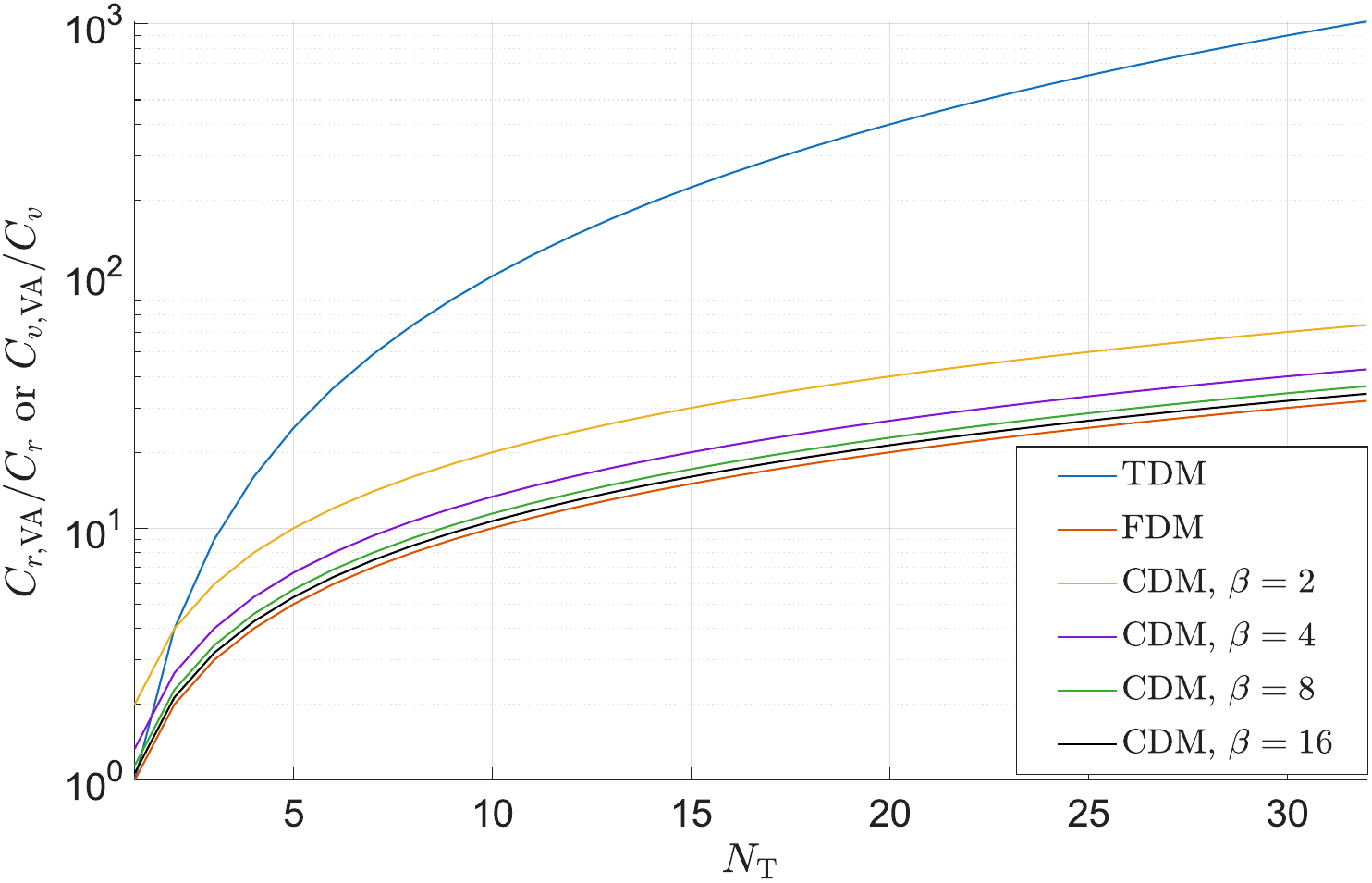}
    \label{fig.VA_tf}}
    
\subfloat[Angle CRLB ratios.]{
    \includegraphics[width=0.95\linewidth]{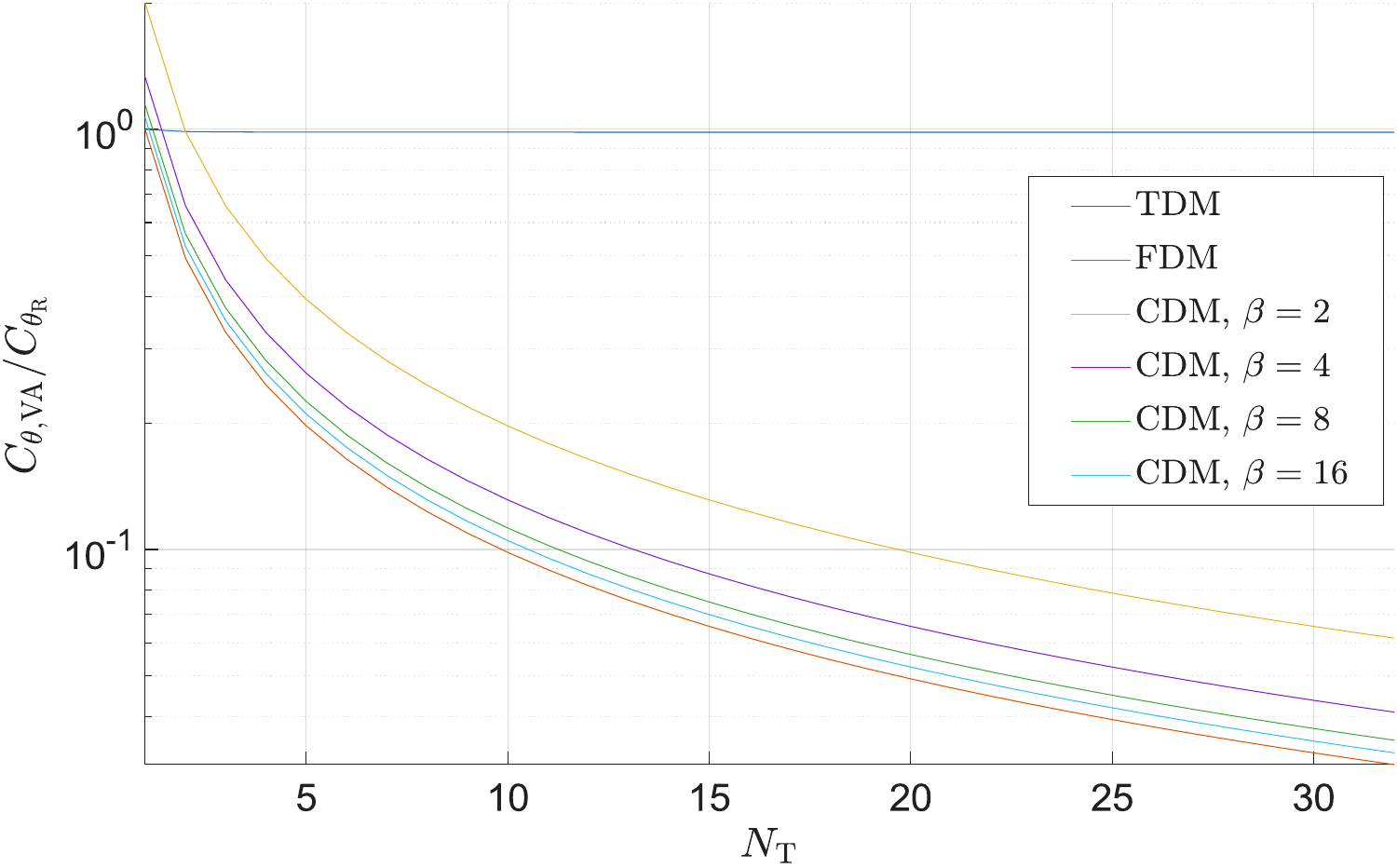}
    \label{fig.VA_a}}
    \caption{CRLB ratios between VA sensing with different multiplexing schemes and non-VA sensing.}
    \label{fig.VA}
\end{figure}

In contrast to range and radial velocity, VA significantly reduces the CRLB of AoA, as shown in Fig.~\ref{fig.VA}(b). While TDM cannot benefit from reduced angle CRLB due to the power issue, the results of FDM and CDM show the same decreasing trend with a slope of $1/N_\text{T}$, and the difference between them is controlled by $\beta$. In conclusion, VA with FDM and CDM provides improvement in the CRLB of AoA at the cost of degraded achievable accuracy of range and radial velocity. In contrast, TDM-VA has range, radial velocity, and AoA CRLBs that are $N_\text{T}$ times higher than those of other multiplexing schemes, and the improvement in AoA CRLB vanishes. However, it is worth noting that CRLB captures only the achievable accuracy, which is one dimension of radar performance, whereas the advantages of VA primarily lie in enhanced angle resolution and target separability, without requiring additional radio-frequency chains.

\section{Conclusion}\label{sec.conclusion}
This paper provides a unified CRLB framework for radar sensing systems. Based on the generic signal model, we analyzed the coupling issue between the delay and Doppler elements in FIM, and presented the conditions under which the coupling can be eliminated or neglected. Afterward, we investigated the fulfillment of the conditions for four representative waveforms, including FMCW, PMCW, OFDM, and OTFS, and derived their CRLBs based on the proposed framework. Moreover, the virtual array with different multiplexing schemes is included in the framework, and its influence on the CRLBs is analyzed. 
The numerical results show a high consistency with those obtained by the waveform-specific analysis in related works, demonstrating the correctness and applicability of the proposed framework. Compared to waveform-specific derivation in related works, the analysis within the unified framework provides a consistent form with a simple expression, enabling a direct performance comparison and offering a tractable metric for optimization designs such as beamforming and resource allocation. The proposed unified framework demonstrates strong generality, flexibility, and waveform-compatibility, offering a novel tool for CRLB analysis of various waveforms.

\begin{appendices}
\section{Derivation of Signal-Level FIM}\label{sec.SLF}
This section provides a summary of the derivation of FIM in previous works \cite{923295,7347470,rs12182913}.
Without loss of generality, the time origin can be chosen such that the energy centroid of the waveform is zero, since this only corresponds to a shift of the time reference and cannot influence the information content. 
Defining energy centroid: 
\begin{align}
    T_0=\frac{1}{E_\text{s}}\int_{-\infty}^{\infty} t|s(t)|^2dt.
\end{align}
Then the time axis is redefined by a shift of $t\leftarrow t-T_0$, i.e., the energy centroid is moved to the origin, hence
\begin{align}
    \int_{-\infty}^{\infty} t|s(t)|^2dt=0.\label{eq.energycent}
\end{align}

For $s(t)$ with constant envelop or \ac{i.i.d.} symbols, the energy is uniformly distributed over time, thus $\int_{-\infty}^{0}t|s(t)|^2dt\approx -\int_{0}^{\infty}t|s(t)|^2dt$, $T_0\approx{T_\text{F}}/{2}$. This holds for ISAC waveforms such as FMCW, PMCW, OFDM, OTFS, etc.

The received signal is reformulated by
\begin{align}
    \mathbf{r}(t)&=Ae^{j\phi}\mathbf{a}_\text{R}^*(\theta_\text{R}){s}(t-\tau)e^{j2\pi f_\text{D}(t+T_0)}+\mathbf{n}(t).
\end{align}

The FIM of $\pmb{\Theta}$ is given by \cite{7347470}
\begin{align}
    \mathbf{F}_{i,j}\negthinspace=\negthinspace-\mathbb{E}\Big\{\frac{\partial^2\log f(\mathbf{r}|\pmb{\Theta})}{\partial{\Theta}_i\partial{\Theta}_j}\Big\}\negthinspace=\negthinspace
    \frac{2}{\sigma^2}\mathcal{R}\Big\{\negthinspace\int_{-\infty}^{\infty}\frac{\partial \pmb{\mu}^H(t)}{\partial\Theta_i}\frac{\partial \pmb{\mu}(t)}{\partial\Theta_j}dt\negthinspace\Big\},
\end{align}
where the derivatives are given by
\begin{align}
    &\frac{\partial \pmb{\mu}(t)}{\partial A}=\frac{\pmb{\mu}(t)}{A},\\
    &\frac{\partial \pmb{\mu}(t)}{\partial \phi}=j\pmb{\mu}(t),\\
    &\frac{\partial \pmb{\mu}(t)}{\partial \tau}=-Ae^{j\phi}\mathbf{a}_\text{R}^*(\theta_\text{R})\dot{s}(t-\tau)e^{j2\pi f_\text{D}(t+T_0)},\\
    &\frac{\partial \pmb{\mu}(t)}{\partial f_\text{D}}=j2\pi (t+T_0)\pmb{\mu}(t),\\
    &\frac{\partial \pmb{\mu}(t)}{\partial \theta_\text{R}}=Ae^{j\phi}\dot{\mathbf{a}}_\text{R}^*(\theta_\text{R})s(t-\tau)e^{j2\pi f_\text{D}(t+T_0)},
\end{align}
where $\dot{\mathbf{a}}^*_\text{R}(\theta_\text{R})=\partial\mathbf{a}^*_\text{R}(\theta_\text{R})/\partial\theta_\text{R}=-j\pi\cos(\theta_\text{R})\mathbf{D}\mathbf{a}^*_\text{R}(\theta_\text{R})$, $\mathbf{D}=\mathrm{diag}(0,1,...,N_\text{R}-1)$, and $\dot{s}(t-\tau)=d s(t-\tau)/d(t-\tau)$.
The elements of $\mathbf{F}$ are calculated as follows:
\begin{align}
    &F_{AA}\negthinspace=\negthinspace\frac{2}{\sigma^2}\negthinspace\mathcal{R}\Big\{\negthinspace\int\negthinspace\frac{1}{A^2}\pmb{\mu}^H\pmb{\mu}dt\Big\}\negthinspace=\negthinspace\frac{2N_\text{R}}{\sigma^2}\negthinspace\int\negthinspace|s(t-\tau)|^2dt\negthinspace=\negthinspace\frac{2N_\text{R}}{\sigma^2} E_\text{s},
\end{align}
where $E_\text{s}=\int|s(t)|^2dt$ denotes the signal energy. %{\color{red}With sampling, the signal energy can be written as $E_\text{s}=P_\text{s}N_\text{s}$, where $N_\text{s}$ is the number of samples.}
\begin{align}
    F_{A\phi}&=\frac{2}{\sigma^2}\mathcal{R}\Big\{\int j\frac{1}{A}\pmb{\mu}^H\pmb{\mu}dt\Big\}=0.
\end{align}
\begin{align}
    F_{A\tau}=-\frac{2}{\sigma^2}N_\text{R}A\mathcal{R}\Big\{\int s^*(t-\tau)\dot{s}(t-\tau)dt\Big\}=0,
\end{align}
where $C_0=\int s^*(t-\tau)\dot{s}(t-\tau)dt$ is pure imaginary \cite{rs12182913} since
\begin{align}
    &\frac{d|s(t-\tau)|^2}{d(t-\tau)}
    \overset{u=t-\tau}{=}\frac{d|s(u)|^2}{du}= \frac{d(s(u)s^*(u))}{du}\nonumber\\
    &= \dot{s}(u)s^*(u) +\dot{s}^*(u)s(u)\nonumber\\
    &= 2\mathcal{R}\{s^*(u)\dot{s}(u)\}=2\mathcal{R}\{s^*(t-\tau)\dot{s}(t-\tau)\},
\end{align}
\begin{align}
    &\mathcal{R}\{C_0\}\negthinspace=\negthinspace\mathcal{R}\Big\{\int_{-\infty}^{\infty}s^*(u)\dot{s}(u)du\Big\}\negthinspace=\negthinspace\negthinspace\int_{-\infty}^{\infty}\frac{1}{2} \frac{d|s(u)|^2}{du}\, du \nonumber\\
    &=\frac{1}{2}|s(u)|^2\Big|_{-\infty}^{+\infty}= 0,
\end{align}
which holds for all time-limited signals.
\begin{align}
    &F_{Af_\text{D}}=\frac{2}{\sigma^2}\mathcal{R}\{j2\pi AN_\text{R}\int (t+T_0)|s(t-\tau)|^2dt\}=0.
\end{align}
\begin{align}
    &F_{A\theta_\text{R}}=\frac{2}{\sigma^2}A\mathcal{R}\Big\{\int \mathbf{a}_\text{R}^T(\theta_\text{R})\dot{\mathbf{a}}_\text{R}^*(\theta_\text{R})|s(t-\tau)|^2dt\Big\}\nonumber\\
    &=\frac{2}{\sigma^2}A\mathcal{R}\Big\{\int  \frac{-j\pi(N_\text{R}-1)N_\text{R}\cos(\theta_\text{R})}{2}|s(t-\tau)|^2dt\Big\}\negthinspace=0.
\end{align}
\begin{align}
    F_{\phi\phi}=%\frac{2}{\sigma^2}\int -j^2\pmb{\mu}^H\pmb{\mu}dt
    \frac{2}{\sigma^2}N_\text{R}A^2\int  |{s}(t-\tau)|^2dt=2N_\text{R}\gamma.
\end{align}
%where $\gamma=\frac{A^2E_\text{s}}{\sigma^2}$ represents the energy \ac{ESNR} calculating the overall energy ratio between signal and AWGN at the Rx. 
\begin{align}
    &F_{\phi\tau}=\frac{2}{\sigma^2}A^2N_\text{R}\mathcal{R}\Big\{j\int  s^*(t-\tau)\dot{s}(t-\tau)dt\Big\}\nonumber\\
    &=\frac{2}{\sigma^2}A^2N_\text{R}\mathcal{R}\{jC_0\}=j\frac{2}{\sigma^2}A^2N_\text{R}C_0.
\end{align}
\begin{align}
    &F_{\phi f_\text{D}}=\frac{4}{\sigma^2}\pi A^2N_\text{R} \int (t+T_0) |s(t-\tau)|^2 dt\nonumber\\
    &\overset{u=t-\tau}{=}\frac{4}{\sigma^2}\pi A^2N_\text{R}\Big(\int u|s(u)|^2du+(T_0+\tau) E_\text{s}\Big)\nonumber\\
    &\overset{(\ref{eq.energycent})}{=}\frac{4}{\sigma^2}\pi A^2N_\text{R}(T_0+\tau) E_\text{s}=4\pi N_\text{R}(T_0+\tau)\gamma.
\end{align}
\begin{align}
    &F_{\phi\theta}=-\frac{2}{\sigma^2}\mathcal{R}\Big\{jA^2\mathbf{a}_\text{R}^T(\theta_\text{R})\dot{\mathbf{a}}^*_\text{R}(\theta_\text{R})\int_{-\infty}^{\infty}|s(t-\tau)|^2dt\Big\}\nonumber\\
    &\negthinspace=\negthinspace-\frac{\pi A^2N_\text{R}(N_\text{R}-1)\cos(\theta_\text{R})}{\sigma^2}E_\text{s}\negthinspace=\negthinspace-\pi N_\text{R}(N_\text{R}-1)\cos(\theta_\text{R})\gamma.
\end{align}
\begin{align}
    F_{\tau\tau}=\frac{2}{\sigma^2}A^2N_\text{R}\int |\dot{s}(t-\tau)|^2dt=8\pi^2N_\text{R}\gamma B_\text{rms}^2,%=\frac{2\pi^2\text{SNR}B^2}{3}\\
\end{align}
where $\dot{s}(t)\leftrightarrow j2\pi fS(f)$, and according to the Parseval's theorem:
\begin{align}
    \negthinspace\int\negthinspace|\dot{s}(t)|\negthinspace^2\negthinspace dt\negthinspace=\negthinspace\int\negthinspace|j2\pi fS(f)|\negthinspace^2\negthinspace d\negthinspace f\negthinspace=\negthinspace4\pi^2\negthinspace\int\negthinspace f^2\negthinspace|S(f)|\negthinspace^2\negthinspace d\negthinspace f\negthinspace=\negthinspace 4\pi^2\negthinspace B_{\negthinspace\text{rms}}^2E_\text{s}.\negthinspace
\end{align}

\begin{align}
    &F_{\tau f_\text{D}}=-\frac{4\pi}{\sigma^2}A^2N_\text{R}\mathcal{R}\Big\{j\int(t+T_0) \dot{s}^*(t-\tau)s(t-\tau)dt\Big\}\nonumber\\
    &=\frac{4\pi}{\sigma^2}A^2N_\text{R}(\mathcal{I}\{C_1\}+(T_0+\tau)jC_0).
    \end{align}
    \begin{align}
    &F_{\tau\theta}=-\frac{2}{\sigma^2}\mathcal{R}\Big\{A^2\mathbf{a}_\text{R}^T(\theta_\text{R})\dot{\mathbf{a}}_\text{R}^*(\theta_\text{R})\int \dot{s}^*(t-\tau)s(t-\tau)dt\Big\}\nonumber\\
    &=\frac{\pi A^2N_\text{R}(N_\text{R}-1)\cos(\theta_\text{R})}{\sigma^2}\mathcal{R}\Big\{j\int \dot{s}^*(t-\tau)s(t-\tau)dt\Big\}\nonumber\\
    &=-j\frac{\pi A^2 N_\text{R}(N_\text{R}-1)\cos(\theta_\text{R})}{\sigma^2}C_0.
\end{align}
\begin{align}
    &F_{f_\text{D}f_\text{D}}=\frac{8\pi^2A^2N_\text{R}}{\sigma^2}\int (t+T_0)^2|s(t-\tau)|^2dt\nonumber\\
    &=\frac{8\pi^2A^2N_\text{R}}{\sigma^2}\int (u+T_0+\tau)^2|s(u)|^2du\nonumber\\
    &=\frac{8\pi^2A^2N_\text{R}}{\sigma^2}\int (u^2+2u(T_0+\tau)+(T_0+\tau)^2)|s(u)|^2du\nonumber\\
    &\overset{(\ref{eq.energycent})}{=}8\pi^2N_\text{R}\gamma(T_\text{rms}^2+(T_0+\tau)^2).
\end{align}
%where the term of $\int 2u(T_0+\tau)|s(u)|^2du=0$ due to (\ref{eq.energycent}). %And the RMS time is defined by
\begin{align}
    &\negthinspace F_{f_\text{D}\theta_\text{R}}\negthinspace=\negthinspace-\frac{2}{\sigma^2}\pi^2A^2N_\text{R}(N_\text{R}\negthinspace-\negthinspace1)\cos(\theta_\text{R}) \negthinspace\int\negthinspace (t\negthinspace+\negthinspace T_0)|s(t\negthinspace-\negthinspace\tau)|^2dt\nonumber\\
    &\overset{(\ref{eq.energycent})}{=}-{2}\pi^2\cos(\theta_\text{R}) (\tau+T_0) \gamma N_\text{R}(N_\text{R}-1).
\end{align}
\begin{align}
    &F_{\theta_\text{R}\theta_\text{R}}
    =\frac{2}{\sigma^2}\pi^2A^2\cos^2(\theta_\text{R})\int|s(t-\tau)|^2dt\sum_{n=0}^{N_\text{R}-1}n^2\nonumber\\
    &=\frac{\pi^2\cos^2(\theta_\text{R})\gamma(N_\text{R}-1)N_\text{R}(2N_\text{R}-1)}{3}.
\end{align}

The amplitude is decoupled from other parameters, thus its corresponding rows and columns can be directly removed from the FIM.
The resulting FIM is then given by
\begin{align}
    \mathbf{F}'
    =\begin{bmatrix}
        F_{\phi\phi} & F_{\phi\tau} & F_{\phi f_\text{D}} & F_{\phi\theta_\text{R}}\\
        F_{\phi\tau} & F_{\tau\tau} & F_{\tau f_\text{D}} & F_{\tau\theta_\text{R}}\\
         F_{\phi f_\text{D}} & F_{\tau f_\text{D}} & F_{f_\text{D}f_\text{D}} & F_{f_\text{D}\theta_\text{R}}\\
        F_{\phi\theta_\text{R}} & F_{\tau\theta_\text{R}} & F_{f_\text{D}\theta_\text{R}} & F_{\theta_\text{R}\theta_\text{R}}
    \end{bmatrix}=\begin{bmatrix}
        F_{\phi\phi} & \mathbf{F}_{\phi s}\\\mathbf{F}_{s\phi}&\mathbf{F}_{ss}
    \end{bmatrix},
\end{align}
The EFIM of the parameters of interest is then calculated by
\begin{align}
    \mathbf{E}=\mathbf{F}_{ss}-\mathbf{F}_{s\phi}F_{\phi\phi}^{-1}\mathbf{F}_{\phi s},
\end{align}
the result is given in (\ref{eq.efim}).

{\textit{Remark:} Since the received signal experiences sampling before further processing, the signal energy can be written as $E_\text{s}=\sum_n |s[n]|^2=P_\text{s}N_\text{s}$, where $N_\text{s}$ is the number of samples. Hence, the relationship between ESNR and SNR can be given by $\gamma=N_\text{s}\,\text{SNR}$. For PMCW, the sampling rate is generally the chip period $T_\text{c}=(1+\alpha)/B$, hence $N_\text{s}=BT_\text{F}/(1+\alpha)$; for OFDM and OTFS, $T_\text{s}=1/B$, $N_\text{s}=BT_\text{F}$; while for FMCW, the sampling rate does not depend on $B$. However, since upsampling or downsampling may be deployed at the radar receiver, $N_\text{s}$ may change accordingly. Therefore, we describe the CRLBs via ESNR instead of SNR in the main text.}

\section{RMS Bandwidth of Different Pulses}\label{sec.brms}

\subsection{Rectangular Pulse}
For rectangular pulses, $g(t)=\text{rect}({t}/{T_\text{c}})$, $G(f)=T_\text{c}\text{sinc}(fT_\text{c})$, while the sidelobes of the spectrum are generally filtered out, resulting in a bandwidth of $B=2/T_\text{c}$. Hence,
\begin{align}
    &B_\text{rms}^2=\frac{\int_{-{1}/{T_\text{c}}}^{{1}/{T_\text{c}}} f^2|T_\text{c}\text{sinc}(fT_\text{c})|^2df}{\int_{-{1}/{T_\text{c}}}^{{1}/{T_\text{c}}}  |T_\text{c}\text{sinc}(fT_\text{c})|^2df}\overset{x=fT_\text{c}}{=}\frac{\int_{-1}^1 x^2\text{sinc}^2(x)dx}{\int_{-1}^1 T_\text{c}^2\text{sinc}^2(x)dx}\nonumber\\
    &=\frac{1}{2\pi T_\text{c}^2\text{Si}(2\pi)}\approx \frac{0.1122}{T_\text{c}^2}\approx0.0281B^2,
\end{align}
where
\begin{align}
    &\int_{-1}^1 x^2\text{sinc}^2(x)dx=\frac{1}{\pi^2}\int_{-1}^1 \sin^2(\pi x)dx=\frac{1}{\pi^2},
\end{align}
\begin{align}
    \int_{-1}^1 \text{sinc}^2(x)dx=\frac{2}{\pi}\int_{0}^\pi\Big(\frac{\sin(u)}{u}\Big)^2du=\frac{2}{\pi}\text{Si}(2\pi),
\end{align}
where $\text{Si}(x)=\int_0^x\frac{\sin(t)}{t}dt$ is the sine integral function \cite{article}, $\text{Si}(2\pi)\approx1.4182$, and  
\begin{align}
    &\negthinspace\text{Si}(2\pi)\negthinspace=\negthinspace\int_{0}^{2\pi}\negthinspace\frac{\sin(t)}{t}\negthinspace dt\negthinspace=\negthinspace\int_{0}^{\pi}\negthinspace\frac{\sin(2u)}{u}\negthinspace du\negthinspace=\negthinspace\int_{0}^{\pi}\negthinspace\frac{2\negthinspace\sin(u)\negthinspace\cos(u)}{u}\negthinspace du\negthinspace\nonumber\\
    &\negthinspace=\negthinspace\frac{\sin^2(u)}{u}\Big|_{0}^\pi+\int_{0}^{\pi}\Big(\frac{\sin(u)}{u}\Big)^2du=\int_{0}^{\pi}\Big(\frac{\sin(u)}{u}\Big)^2du.
\end{align}

\subsection{Sinc Pulse}
For sinc pulses, $g(t)=\text{sinc}(t/T_\text{c})$, $G(f)=T_\text{c}\mathrm{rect}(fT_\text{c})$, $B=1/T_\text{c}$, thus the RMS bandwidth is calculated by
\begin{align}
    &B_\text{rms}^2=\frac{\int f^2|T_\text{c}\text{rect}(fT_\text{c})|^2df}{\int |T_\text{c}\text{rect}(fT_\text{c})|^2df}\overset{x=fT_\text{c}}{=}\frac{\int x^2|\text{rect}(x)|^2dx}{\int T_\text{c}^2|\text{rect}(x)|^2dx}\nonumber\\
    &=\frac{1}{12T_\text{c}^2}=\frac{B^2}{12}.
\end{align}

\subsection{RRC Pulse}
RRC and RC filters are used in communication systems for pulse shaping and matched filtering, limiting the signal bandwidth into $B={(1+\alpha)}/{T_\text{c}}$ and avoiding \ac{ISI}. The power spectrum of a RRC pulse is given by
\begin{align}
    %g(t)=\begin{cases}
    %    \frac{\pi}{4T_\text{c}}\mathrm{sinc}\Big(\frac{1}{2\alpha}\Big), & t=\pm \frac{{T}_\text{c}}{2\alpha}\\
    %    \frac{1}{T_\text{c}}\mathrm{sinc}\Big(\frac{t}{T_\text{c}}\Big)\frac{\cos\Big(\frac{\pi\alpha t}{T_\text{c}}\Big)}{1-\Big(\frac{2\alpha t}{T_\text{c}}\Big)^2}, &\text{otherwise.}
    %\end{cases}\\
    |G(f)|^2=\begin{cases}
        T_\text{c},&|f|\leq f_\text{L},\\
        \frac{T_\text{c}}{2}\Big[1+\cos\Big(\frac{\pi T_\text{c}}{\alpha}(|f|-f_\text{L})\Big)\Big],&f_\text{L}<|f|\le f_\text{H},\\
        0,&\text{else},
    \end{cases}\label{eq.rrc}
\end{align}
where $\alpha$ is the roll-off factor, $f_\text{L}=\frac{1-\alpha}{2T_\text{c}}$, and $f_\text{H}=\frac{1+\alpha}{2T_\text{c}}$. The definition of RRC spectrum in (\ref{eq.rrc}) considers the normalized power, i.e., $\int_{-\infty}^{\infty}|G(f)|^2df=1$. The $B_\text{rms}^2$ is derived by
\begin{align}
    &B_\text{rms}^2=\int_{-\infty}^{\infty}f^2|G(f)|^2df\nonumber\\
    &\negthinspace=\negthinspace2\Big(\negthinspace\int_{0}^{f_\text{L}}\negthinspace f^2T_\text{c}df+\frac{T_\text{c}}{2}\negthinspace\int_{f_\text{L}}^{f_\text{H}}\negthinspace f^2\Big[1\negthinspace+\negthinspace\cos\Big(\frac{\pi T_\text{c}}{\alpha}(f-f_\text{L})\Big)\Big]df\negthinspace\Big)\nonumber\\
    &\negthinspace=\negthinspace\frac{(1-\alpha)^3}{12T_\text{c}^2}\negthinspace+\negthinspace\frac{\alpha^3+3\alpha}{12T_\text{c}^2}\negthinspace+\negthinspace
    T_\text{c}\negthinspace\int_{0}^{f_\text{H}-f_\text{L}}\negthinspace(f+f_\text{L})^2\cos\Big(\frac{\pi T_\text{c}}{\alpha}f\Big)df\nonumber\\
    &\overset{u=\pi T_\text{c}f/\alpha}{=}\frac{3\alpha^2+1}{12T_\text{c}^2}+\frac{\alpha}{\pi}\int_{0}^{\pi}\big(\frac{\alpha}{\pi T_\text{c}}u+f_\text{L}\big)^2\cos(u)du\nonumber\\
    &\overset{(\ref{eq.where1})}{=}\frac{3\alpha^2+1}{12T_\text{c}^2}+\frac{\alpha}{\pi}(-\frac{4\alpha f_\text{L}}{\pi T_\text{c}}-\frac{2\alpha^2}{\pi T_\text{c}^2})\nonumber\\
    &=\frac{3\alpha^2+1}{12T_\text{c}^2}-\frac{2\alpha^2}{\pi^2T_\text{c}^2}=\frac{B^2}{12}\cdot\frac{(3\pi^2-24)\alpha^2+\pi^2}{\pi^2(1+\alpha)^2},
\end{align}
where
\begin{align}
    \negthinspace\int_{0}^{\negthinspace\pi}\negthinspace(\negthinspace\frac{\alpha}{\pi T_\text{c}}\negthinspace u\negthinspace+\negthinspace f_\text{L}\negthinspace)\negthinspace^2\negthinspace\cos(\negthinspace u\negthinspace)du\negthinspace=\negthinspace\int_{0}^{\negthinspace\pi}\negthinspace(f_\text{L}^2\negthinspace+\negthinspace\frac{2\alpha f_\text{L}}{\pi T_\text{c}}\negthinspace u\negthinspace+\negthinspace\frac{\alpha^2}{\pi^2T_\text{c}^2}\negthinspace u^2)\negthinspace\cos(\negthinspace u\negthinspace)du,\negthinspace\label{eq.where1}
\end{align}
and the individual terms are calculated by
\begin{align}
    &\int_{0}^{\pi}f_\text{L}^2\cos(u)du=0,\\
    &\frac{2\alpha f_\text{L}}{\pi T_\text{c}}\int_{0}^{\pi}u\cos(u)du=-\frac{4\alpha f_\text{L}}{\pi T_\text{c}},\\
    &\frac{\alpha^2}{\pi^2T_\text{c}^2}\int_{0}^{\pi}u^2\cos(u)du=-\frac{2\alpha^2}{\pi T_\text{c}^2}.
\end{align}

\subsection{RC Pulse}
The spectrum of RC pulses is the square of that of RRC, thus
\begin{align}
    |G(f)|^2=\begin{cases}
        T_\text{c}^2,&|f|\leq f_\text{L}\\
        \frac{T_\text{c}^2}{4}\Big[1+\cos\Big(\frac{\pi T_\text{c}}{\alpha}(|f|-f_\text{L})\Big)\Big]^2,&f_\text{L}<|f|\le f_\text{H}\\
        0,&\text{else}.
    \end{cases}\label{eq.rc}
\end{align}
The denominator and numerator of $B_\text{rms}^2$ are derived by
\begin{align}
    &\int\negthinspace |G(f)|^2 df\negthinspace=\negthinspace2\negthinspace\int_0^{f_\text{L}}\negthinspace T_\text{c}^2\negthinspace df\negthinspace+\negthinspace2\negthinspace\int_{f_\text{L}}^{f_\text{H}}\negthinspace\frac{T_\text{c}^2}{4}\negthinspace\Big(\negthinspace1\negthinspace+\negthinspace\cos\negthinspace\big(\negthinspace\frac{\pi T\negthinspace_\text{c}}{\alpha}(\negthinspace f\negthinspace-\negthinspace f_\text{L}\negthinspace)\negthinspace\big)\negthinspace\Big)^2\negthinspace df\negthinspace\nonumber\\
    &=\negthinspace2f_\text{L}T_\text{c}^2\negthinspace+\negthinspace\frac{T_\text{c}^2}{2}\negthinspace\int_{0}^{f_\text{H}-f_\text{L}}\negthinspace\Big(\negthinspace1\negthinspace+\negthinspace2\negthinspace\cos\negthinspace\big(\negthinspace\frac{\pi T_\text{c}}{\alpha}\negthinspace f\big)\negthinspace+\negthinspace\frac{1\negthinspace+\negthinspace\cos\negthinspace\big(\negthinspace\frac{2\pi T_\text{c}}{\alpha}\negthinspace f\big)}{2}\negthinspace\Big)\negthinspace df\negthinspace\nonumber\\
    &=\negthinspace(\negthinspace1\negthinspace-\negthinspace\alpha\negthinspace)T\negthinspace_\text{c}\negthinspace+\negthinspace\frac{3T\negthinspace_\text{c}\alpha}{4}\negthinspace+\negthinspace T_\text{c}^2\negthinspace\int_{0}^{f_\text{H}\negthinspace-\negthinspace f_\text{L}}\negthinspace\Big(\negthinspace\cos\negthinspace\big(\negthinspace\frac{\pi T\negthinspace_\text{c}}{\alpha}\negthinspace f\big)\negthinspace+\negthinspace\frac{1}{4}\negthinspace\cos\negthinspace\big(\negthinspace\frac{2\pi T\negthinspace_\text{c}}{\alpha}\negthinspace f\big)\negthinspace\Big)\negthinspace df\negthinspace\nonumber\\
    &=(1\negthinspace-\negthinspace\alpha)T_\text{c}\negthinspace+\negthinspace\frac{3T_\text{c}\alpha}{4}\negthinspace+\negthinspace\frac{\alpha T_\text{c}}{\pi}\Big(\negthinspace\int_0^\pi\negthinspace\cos(u)du\negthinspace+\negthinspace\frac{1}{8}\negthinspace\int_0^{2\pi}\negthinspace\cos(u)du\negthinspace\Big)\nonumber\\
    &=T_\text{c}(1-\frac{\alpha}{4}),
\end{align}
\begin{align}
    &\int\negthinspace f^2|G(f)|^2 df\nonumber\\
    &=2\int_0^{f_\text{L}}T_\text{c}^2f^2df+\frac{T_\text{c}^2}{2}\int_{f_\text{L}}^{f_\text{H}}f^2\Big(1+\cos\big(\frac{\pi T_\text{c}}{\alpha}(f-f_\text{L})\big)\Big)^2df\nonumber\\
    &\overset{(\ref{eq.104})}{=}\frac{2}{3}T_\text{c}^2f_\text{L}^3+\frac{T_\text{c}^2}{2}\Big(\frac{\alpha^3+3\alpha}{8T_\text{c}^3}-\frac{4\alpha^2}{\pi^2T_\text{c}^3}+\frac{\alpha^3}{4\pi^2T_\text{c}^3}\Big)\nonumber\\
    %&=\frac{2}{3}T_\text{c}^2f_\text{L}^3+\frac{T_\text{c}^2}{2}\int_{f_\text{L}}^{f_\text{H}}f^2\Big(1+2\cos\big(\frac{\pi T_\text{c}}{\alpha}(f-f_\text{L})\big)+\frac{1+\cos\big(\frac{2\pi T_\text{c}}{\alpha}(f-f_\text{L})\big)}{2}\Big)df\nonumber\\
    %&=\frac{(1-\alpha)^3}{12T_\text{c}}+\frac{\alpha^3+3\alpha}{16T_\text{c}}+T_\text{c}^2\int_{0}^{\alpha/T_\text{c}}(f+f_\text{L})^2\Big(\cos\big(\frac{\pi T_\text{c}}{\alpha}f\big)+\frac{1}{4}\cos\big(\frac{2\pi T_\text{c}}{\alpha}f\big)\Big)df\nonumber\\
    %&=\frac{(1-\alpha)^3}{12T_\text{c}}+\frac{\alpha^3+3\alpha}{16T_\text{c}}+\frac{\alpha T_\text{c}}{\pi}\Big(\int_0^\pi (\frac{\alpha}{\pi T_\text{c}}u+f_\text{L})^2\cos(u)du+\frac{1}{8}\int_0^{2\pi}(\frac{\alpha}{2\pi T_\text{c}}u+f_\text{L})^2\cos(u)du\Big)\nonumber\\
    %&=\frac{(1-\alpha)^3}{12T_\text{c}}+\frac{\alpha^3+3\alpha}{16T_\text{c}}-\frac{2\alpha^2}{\pi^2T_\text{c}}+\frac{\alpha^3}{8\pi^2T_\text{c}}\nonumber\\
    &=\frac{(6-\pi^2)\alpha^3+(12\pi^2-96)\alpha^2-3\pi^2\alpha+4\pi^2}{48\pi^2T_\text{c}},
\end{align}
where 
\begin{align}\label{eq.104}
    &\int_{f_\text{L}}^{f_\text{H}}f^2\Big(1+\cos\big(\frac{\pi T_\text{c}}{\alpha}(f-f_\text{L})\big)\Big)^2df\nonumber\\
    &=\negthinspace\int_{f_\text{L}}^{f_\text{H}}\negthinspace f^2\negthinspace\Big(\frac{3}{2}\negthinspace+\negthinspace2\cos\negthinspace\big(\frac{\pi T_\text{c}}{\alpha}(f\negthinspace-\negthinspace f_\text{L})\big)\negthinspace+\negthinspace{\frac{1}{2}\cos\negthinspace\big(\frac{2\pi T_\text{c}}{\alpha}(f\negthinspace-\negthinspace f_\text{L})\big)}\negthinspace\Big)df,
\end{align}
and the individual terms are derived by
\begin{align}
    &\int_{f_\text{L}}^{f_\text{H}}\frac{3}{2}f^2df=\frac{\alpha^3+3\alpha}{8T_\text{c}^3},\\
    &\int_{\negthinspace f_\text{L}}^{\negthinspace f_\text{H}}\negthinspace 2\negthinspace f^2\negthinspace\cos\negthinspace\big(\frac{\pi T_\text{c}}{\alpha}(f\negthinspace-\negthinspace f_\text{L})\negthinspace\big)df\negthinspace=\negthinspace2\negthinspace\int_{\negthinspace 0}^{\negthinspace f_\text{H}-f_\text{L}}\negthinspace(f+f_\text{L})^2\negthinspace\cos\negthinspace\big(\frac{\pi T_\text{c}}{\alpha}\negthinspace f\big)df\nonumber\\
    &=\frac{2\alpha}{\pi T_\text{c}}\int_0^\pi (\frac{\alpha}{\pi T_\text{c}}u+f_\text{L})^2\cos(u)du\overset{(\ref{eq.where1})}{=}-\frac{4\alpha^2}{\pi^2T_\text{c}^3},\\
    &\int_{\negthinspace f_\text{L}}^{\negthinspace f_\text{H}}\negthinspace \frac{1}{2}\negthinspace f^2\negthinspace\cos\negthinspace\big(\frac{2\pi T_\text{c}}{\alpha}\negthinspace(f\negthinspace-\negthinspace f_\text{L})\big)\negthinspace=\negthinspace\frac{1}{2}\negthinspace\int_{\negthinspace0}^{\negthinspace f_\text{H}-f_\text{L}}\negthinspace(f\negthinspace+\negthinspace f_\text{L})^2\negthinspace\cos\negthinspace\big(\frac{2\pi T_\text{c}}{\alpha}\negthinspace f\big)df\negthinspace\nonumber\\
    &=\frac{\alpha}{4\pi T_\text{c}}\int_0^{2\pi} (\frac{\alpha}{2\pi T_\text{c}}u+f_\text{L})^2\cos(u)du=\frac{\alpha^3}{4\pi^2T_\text{c}^3}.
\end{align}

Consequently, the RMS bandwidth of RC pulses is derived by:
\begin{align}
    &B_\text{rms}^2=\frac{(6-\pi^2)\alpha^3+(12\pi^2-96)\alpha^2-3\pi^2\alpha+4\pi^2}{(48-12\alpha)\pi^2T_\text{c}^2}\nonumber\\
    &=\frac{B^2}{12}\cdot\frac{(6-\pi^2)\alpha^3+(12\pi^2-96)\alpha^2-3\pi^2\alpha+4\pi^2}{\pi^2(4-\alpha)(1+\alpha)^2}.
\end{align}

\end{appendices}

\printbibliography

\vspace{-3em}

\end{document}